\documentclass[a4paper,twocolumn,11pt,longbibliography,accepted=2025-05-25]{quantumarticle}
\pdfoutput=1
\usepackage[utf8]{inputenc}
\usepackage[english]{babel}
\usepackage[T1]{fontenc}
\usepackage{amsmath}

\usepackage[numbers,sort&compress]{natbib}

\usepackage{hyperref}

\usepackage{units}
\usepackage{mathtools} 
\usepackage{amsmath}
\usepackage{amsthm}
\usepackage{amssymb}
\usepackage{graphicx}
\usepackage{color}
\usepackage{xcolor}
\usepackage{bbold}
\usepackage{titlesec}

\definecolor{myurlcolor}{rgb}{0,0,0.7}
\definecolor{myrefcolor}{rgb}{0.1,0,0.9}

\newcommand{\Gen}[0]{\mathsf{Gen}}

\newcommand{\bit}{\{0,1\}}

\newcommand{\SM}{Appendix}

\newcommand{\revA}[1]{{#1}}

\newcommand{\revB}[1]{{#1}}

\newcommand{\revC}[1]{{#1}}

\newcommand{\challenger}[0]{\mathcal{C}}
\newcommand{\adversary}[0]{\mathcal{A}}

\newtheorem{theorem}{Theorem}

\usepackage[linesnumbered, ruled,vlined]{algorithm2e}

\graphicspath{{./images/},{./imagesAppendix/}}

\renewcommand{\eqref}[1]{Eq.~(\ref{#1})} %

\def\app#1#2{%
  \mathrel{%
    \setbox0=\hbox{$#1\sim$}%
    \setbox2=\hbox{%
      \rlap{\hbox{$#1\propto$}}%
      \lower1.1\ht0\box0%
    }%
    \raise0.25\ht2\box2%
  }%
}

\ifx\proof\undefined
\newenvironment{proof}[1][\protect\proofname]{\par
	\normalfont\topsep6\p@\@plus6\p@\relax
	\trivlist
	\itemindent\parindent
	\item[\hskip\labelsep\scshape #1]\ignorespaces
}{%
	\endtrivlist\@endpefalse
}
\providecommand{\proofname}{Proof}
\fi
\newtheorem{claim}{\protect\claimname}

\newcommand{\bra}[1]{\langle #1|}
\newcommand{\ket}[1]{|#1 \rangle}
\makeatother
\newcommand{\braket}[2]{\langle #1 \vert #2 \rangle}
\newcommand{\ketbra}[2]{|#1 \rangle\!\langle #2 |}
\newcommand{\abs}[1]{\left|#1\right|}

\newcommand{\tr}{\mathrm{tr}}

\newcommand\numberthis{\addtocounter{equation}{1}\tag{\theequation}}
\providecommand{\factname}{Fact}
\providecommand{\theoremname}{Theorem}
\providecommand{\claimname}{Claim}
\providecommand{\lemmaname}{Lemma}
\providecommand{\definitionname}{Definition}
\providecommand{\corollaryname}{Corollary}

\newtheorem{corollary}{\protect\corollaryname}

\def\d{\mathrm{d}}

\newcommand{\subfigimg}[3][,]{%
	\setbox1=\hbox{\includegraphics[#1]{#3}}%
	\leavevmode\rlap{\usebox1}%
	\rlap{\hspace*{2pt}\raisebox{\dimexpr\ht1-0.5\baselineskip}{{\bfseries \large\textsf{#2}}}}%
	\phantom{\usebox1}%
}

\definecolor{KB}{rgb}{0.4,0.3,0.9}

\definecolor{THc}{rgb}{0.9,0.3,0.2}

\newcommand{\sectionMain}[1]{
\let\oldaddcontentsline\addcontentsline%
\renewcommand{\addcontentsline}[3]{}%
\section{#1}
\let\addcontentsline\oldaddcontentsline
}

\newcommand{\prlsection}[1]{{\section{#1}}}
\newcommand{\emphsection}[1]{{\em {#1}.---}}

\newcommand{\eq}[1]{\hyperref[eq:#1]{(\ref*{eq:#1})}}

\newcommand{\ihpc}{Institute of High Performance Computing (IHPC), Agency for Science, Technology and Research (A*STAR), 1 Fusionopolis Way, $\#$16-16 Connexis, Singapore 138632, Republic of Singapore\looseness=-1}
\newcommand{\sutd}{Science, Mathematics and Technology Cluster, Singapore University of Technology and Design, 8 Somapah Road, Singapore 487372, Singapore\looseness=-1}
\newcommand{\qinc}{Quantum Innovation Centre (Q.InC), Agency for Science, Technology and Research (A*STAR), 2 Fusionopolis Way, Innovis \#08-03, Singapore 138634, Republic of Singapore\looseness=-1}

\allowdisplaybreaks

\begin{document}

\title{Pseudorandom unitaries are neither real nor sparse nor noise-robust}

\author{Tobias Haug}
\email{tobias.haug@u.nus.edu}
\affiliation{Quantum Research Center, Technology Innovation Institute, Abu Dhabi, UAE}
\affiliation{Blackett Laboratory, Imperial College London SW7 2AZ, UK}

\author{Kishor Bharti}
\email{kishor.bharti1@gmail.com}
\affiliation{\qinc}
\affiliation{\ihpc}

\author{Dax Enshan Koh}
\email{dax\_koh@ihpc.a-star.edu.sg}
\affiliation{\qinc}
\affiliation{\ihpc}
\affiliation{\sutd}

\begin{abstract}
Pseudorandom quantum states (PRSs) and pseudorandom unitaries (PRUs) possess the dual nature of being efficiently constructible while appearing completely random to any efficient quantum algorithm. In this study, we establish fundamental bounds on pseudorandomness. We show that PRSs and PRUs exist only when the probability that an error occurs is negligible, ruling out their generation on noisy intermediate-scale and early fault-tolerant quantum computers. 
Further, we show that PRUs need imaginarity while PRS do not have this restriction. This implies that quantum randomness requires in general a complex-valued formalism of quantum mechanics, while for random quantum states real numbers suffice.
Additionally, we derive lower bounds on the coherence of PRSs and PRUs, ruling out the existence of sparse PRUs and PRSs. 
We also show that the notions of PRS, PRUs and pseudorandom scramblers (PRSSs) are distinct in terms of resource requirements. 
We introduce the concept of pseudoresources, where states which contain a low amount of a given resource masquerade as high-resource states. 
We define pseudocoherence, pseudopurity and pseudoimaginarity, and identify three distinct types of pseudoresources in terms of their masquerading capabilities.
Our work also establishes rigorous bounds on the efficiency of property testing, demonstrating the exponential complexity in distinguishing real quantum states from imaginary ones, in contrast to the efficient measurability of unitary imaginarity. Further, we show an exponential advantage in imaginarity testing when having access to the complex conjugate of the state.
Lastly, we show that the transformation from a complex to a real model of quantum computation is inefficient, in contrast to the reverse process, which is efficient. 
Our results establish fundamental limits on property testing and provide valuable insights into quantum pseudorandomness.
\end{abstract}

\maketitle

 \let\oldaddcontentsline\addcontentsline%
\renewcommand{\addcontentsline}[3]{}%

Randomness is a key resource for quantum information processing~\cite{ji2018pseudorandom,emerson2003pseudo,brandao2016efficient,brandao2016local,nakata2017efficient,scarani2010guaranteed,law2014quantum,herrero2017quantum,acin2016certified,bera2017randomness,calude2008quantum}. However, a simple counting argument reveals that preparing completely random quantum states is computationally intractable. Thus, a core challenge lies in identifying quantum states and unitaries that can be efficiently prepared while remaining indistinguishable from truly random ones.

To solve this problem, the concepts of pseudorandom quantum states (PRSs), pseudorandom unitaries (PRUs)~\cite{ji2018pseudorandom,schuster2024randomunitariesextremelylow,ma2024construct} and pseudorandom scramblers (PRSSs)~\cite{lu2023quantum,ananth2024pseudorandom} have been introduced. There is no efficient quantum algorithm that can distinguish Haar random states and unitaries from PRSs and PRUs, respectively, yet we can efficiently prepare them. 
Conveniently, PRSs can be generated with much shorter circuits compared to state designs~\cite{ji2018pseudorandom,brakerski2019pseudo,bouland2022quantum} and have been shown to be a weaker primitive than one-way functions~\cite{kretschmer2021quantum}. They have also been used to design cryptographic functionalities such as commitment schemes~\cite{ananth2022cryptography,ananth2022pseudorandom,morimae2022quantum}, pseudo one-time pads~\cite{ananth2022cryptography,ananth2022pseudorandom,morimae2022quantum} and pseudorandom strings~\cite{ananth2023pseudorandom}. 
There is an ongoing search to identify PRSs with as simple a structure as possible. Recent work showed that PRSs can have logarithmic entanglement~\cite{bouland2022quantum} and require at least a linear number of non-Clifford gates~\cite{grewal2023improved}.
However, for most other properties, the minimal requirements to generate PRSs, PRSSs and PRUs are unknown.

Imaginarity and coherence are key resources in quantum information processing, playing crucial roles in achieving a quantum advantage in various tasks. Imaginarity characterizes the degree to which states and operations require non-real numbers for their description. Surprising differences between real and complex descriptions of quantum mechanics can arise~\cite{wu2021resource,kondra2022real,chen2022ruling} and there are tasks that can be accomplished only with operations described by non-real numbers~\cite{renou2021quantum,wu2021operational,li2022testing}.
On the other hand, coherence, which measures the degree of being in a superposition, is a fundamental resource in quantum processing~\cite{baumgratz2014quantifying,streltsov2017colloquium,bu2017cohering,zanardi2017coherence,bu2017maximum,bu2017note} that is intricately linked to entanglement~\cite{streltsov2015measuring,chitambar2016relating}, circuit complexity~\cite{bu2022complexity} and serves as a vital ingredient in quantum algorithms~\cite{hillery2016coherence,matera2016coherent,shi2017coherence,wang2021minimizing}. 

Given their importance, can we verify whether coherence or imaginarity are actually present in a quantum state or unitary? For such tasks, the concept of property testing 
has been developed~\cite{rubinfeld1996robust,goldreich1998property,buhrman2008quantum}. This concept has been studied extensively for other properties~\cite{montanaro2013survey,harrow2013testing,gross2021schur,chen2023testing,she2022unitary,brakerski2021unitary,soleimanifar2022testing,bouland2022quantum}.

Here, we provide fundamental bounds on PRSs, PRSSs, PRUs, as well as testing for imaginarity and coherence as a function of the number of qubits $n$. We show that the number of copies needed to test the imaginarity of states scales as $\Omega(2^{n/2})$ with the number of qubits $n$, while for unitaries we require only $O(1)$ queries. \revB{Further, we show an exponential advantage in imaginarity testing with access to the complex conjugate of states.} For coherence, we provide bounds on the number of copies needed for testing as a function of coherence of the state and unitary. This lower bounds the coherence of PRUs as $\omega(\log(n))$ and their imaginarity as $1-\text{negl}(n)$. With our results, we rule out the existence of real and sparse PRUs. Further, we show that PRSs and PRUs are not robust to noise, with a maximal depolarizing noise probability $p=\text{negl}(n)$. We introduce pseudoresource ensembles that mimic states with high resource value using states with a only a low amount of the given resource, and identify three classes of pseudoresources.
\revA{We also show that unitaries that prepare PRS, PRSS and PRUs are separated in resource complexity, proving a clear separation between the three notions of pseudorandomness.}
Finally, we apply our results to show that there is no efficient way to transform from a complex mode of quantum computing to a real one. Our work establishes rigorous limits on property testing and the generation of pseudorandomness.

\prlsection{Preliminaries}%
\emph{Pseudorandom state (PRS):} An ensemble of quantum states indexed by a key (commonly referred to as a keyed family of states~\cite{ji2018pseudorandom}) is considered pseudorandom if it can be generated efficiently and
appears statistically indistinguishable from Haar random states to an observer with limited computational resources. More precisely, a family of PRSs is a set of \revC{pure}
states $\{\ket{\phi_k}\}_{k \in \mathcal{K}}$ of some key space $\mathcal{K}$  which satisfies two conditions: First, they can be efficiently prepared on quantum computers. Second, given $t=\text{poly}(n)$ copies of the state, they are indistinguishable from Haar random quantum states $S(\mathcal{H})$ of an $n$-qubit Hilbert space $\mathcal{H}$, for any efficient quantum algorithm~\cite{ji2018pseudorandom}. \revC{The notion of pseudorandom states was recently generalized to mixed states~\cite{bansal2024pseudorandom}, which we do not address in this work.}

\emph{Pseudorandom unitary (PRU):} Similarly, one can also define PRUs~\cite{ji2018pseudorandom,schuster2024randomunitariesextremelylow,ma2024construct}. A family of unitary operators is considered pseudorandom if it is both efficiently computable and indistinguishable from Haar random unitaries for an observer with bounded computational power. 
\revC{Moreover, variants of PRUs with restricted security have been proposed~\cite{chen2024efficient,metger2024pseudorandom,brakerski2024real,metger2024simple}.}

\emph{Pseudorandom scrambler (PRSS):} Recently, the notion of PRSS was introduced, which lies in between PRS and PRU: PRSS are an ensemble of unitaries that applied to an arbitrary, fixed state generate PRSs~\cite{lu2023quantum,ananth2024pseudorandom}. \revC{It is straightforward to see that any PRU must also be a PRSS, though it had hitherto been unclear whether all PRSSs are also PRU~\cite{lu2023quantum,ananth2024pseudorandom}. In this work, we provide an explicit counterexample of a PRSS that is not a PRU.}

\revC{
\emph{Quantum resources:} 
In quantum information, quantum resources can be regarded as the expensive fuel needed to carry out quantum processing tasks~\cite{chitambar2019quantum}. For example, a key resource is entanglement, which characterizes the non-local correlations in quantum systems. In fault-tolerant quantum computing, the resource theory of magic plays a significant role~\cite{veitch2014resource}. Furthermore, quantum information processing often requires coherence, which describes the degree to which a state exists as a superposition of computational basis states~\cite{baumgratz2014quantifying,streltsov2017colloquium}. Finally, the degree to which complex numbers are needed to describe states is quantified by the resource theory of imaginarity~\cite{wu2021resource}.

Given a state $\rho$, one characterizes the resource content via quantum resource monotones $Q(\rho)$, which satisfy $Q(\rho)\geq0$~\cite{chitambar2019quantum}. These resource monotones are defined by free operations $F_Q$ that cannot increase the resource, i.e.\ $Q(F_Q(\rho))\leq Q(\rho)$, and free states with $Q(\rho)=0$. While free operations are easy to implement, performing non-trivial tasks within the resource theory typically requires resource states with $Q(\rho)>0$. 
}

\emph{Property Testing:} Property testing checks whether a given state has a particular property or resource $Q$. In particular, given $t$ copies of $\ket{\psi}$, we want to determine whether $\ket{\psi}$ has $Q(\ket{\psi})\le \beta$, where $\beta$ is the threshold for the property and $Q(\ket{\psi})$ is a resource measure with $Q(\ket{\psi})\ge0$. The property tester accepts when $Q(\ket{\psi})\le \beta$, and rejects when $Q(\ket{\psi})\ge \delta$, where $\delta$ is the rejection threshold, which is chosen as $\delta>\beta$.

For further details and formal definitions on cryptography and property testing, refer to the \SM~\ref{sec:basics} and~\ref{sec:propert_test}.%

\begin{table}[htbp]\centering
\begin{tabular}{| c|c|c| }
    \hline
  Resource $Q$  & Pure states
            & Unitaries\\
    \hline
    Purity   &   $O(1)$~\cite{barenco1997stabilization}   &   $O(1)$       \\
\hline
Entanglement   &    $O(1)$~\cite{ekert2002direct}  &    $O(1)$~\cite{harrow2013testing}     \\
    \hline
Magic   &   $O(1)$~\cite{gross2021schur}    &   $O(1)$~\cite{gross2021schur}       \\
    \hline
Coherence   &   $O(1)$ [Prop.~\ref{prop:test_coh_low_sup}]   &   $O(1)$~\cite{zanardi2017coherence}       \\
    \hline
Imaginarity   &    $\Omega(2^{n/2})$ [Thm.~\ref{thm:no_test_imag}]  &   $O(1)$~\cite{huang2021demonstrating}\\
    \hline
\end{tabular}
\caption{Copies $M$ needed for property tester $\mathcal{A}_Q$ for the resources of imaginarity, entanglement, magic, coherence and purity, where we regard pure states and unitaries (except for purity where we regard mixed states and channels). Here, we assume $\beta=0$, $\delta=\Omega(1)$ for the tester and choose a representative measure with the best known protocol for each property. Further details in \SM~\ref{sec:comparison}.}
\label{tab:comparison}
\end{table}

\prlsection{Noise} 
PRSs and PRUs physically prepared on quantum computers are subject to noise originating from the environment or imperfect experimental control. One can characterize the amount of noise via the  property of purity $\gamma(\rho)=\text{tr}(\rho^2)$ of a state $\rho$~\cite{nielsen2002quantum}, where $\gamma=1$ for any pure state and $\gamma=2^{-n}$ for maximally mixed states. It has been shown that PRSs must have maximal purity up to a negligible perturbation~\cite{ananth2022pseudorandom}.
As we now show, PRSs and PRUs are very fragile with respect to noise channels:
\begin{theorem}[PRSs and PRUs are not robust to noise (\SM~\ref{sec:noise})]\label{thm:noise}
Any ensemble of unitaries or states subject to a single-qubit depolarizing noise channel $\Lambda(X)=(1-p) X+p\tr(X)I_1/2$ can be PRSs and PRUs only with at most $p=\mathrm{negl}(n)$ depolarizing probability.
\end{theorem}
Theorem~\ref{thm:noise} is proven via the SWAP test~\cite{barenco1997stabilization}. In particular, any state or unitary subject to a local depolarizing channel with $p=\Omega(1/\text{poly}(n))$ noise probability can be efficiently distinguished from Haar random states or unitaries.

\prlsection{Imaginarity of states}%
Next, we discuss the property of imaginarity. An $n$-qubit quantum state $\rho$ is non-imaginary or real when all the coefficients in the expansion
\begin{equation}\label{eq:state}
\rho=\sum_{k,\ell} \rho_{k,\ell}\ket{k}\bra{\ell}\,
\end{equation}
are real, i.e.~$\rho_{k,\ell}\in \mathbb{R}$ for all $k,\ell$.
Note that here we do not allow for flag qubits that serve as an indicator of imaginarity as used in the rebit model of quantum computation~\cite{rudolph20022}.
Further note that imaginarity depends on the chosen canonical basis of the system. For example, redefining the computational basis from $\ket{0},\ket{1}$ to $\ket{+y},\ket{-y}$, where $\ket{\pm y}=\frac{1}{\sqrt{2}}(\ket{0}\pm i\ket{1})$,  changes a real state into an imaginary one. 
Here, we make the usual choice to define imaginarity with respect to the canonical computational basis $\ket{0}$ and $\ket{1}$.

The degree to which a state $\rho$ requires complex numbers can be quantified using measures of imaginarity such as the robustness of imaginarity~\cite{hickey2018quantifying,wu2021resource}
\begin{equation}\label{eq:imag_robust}
    R(\rho)=\min_{\tau}\left\{s\ge0 : \frac{\rho+s\tau}{1+s}\, \text{real}\right\}= \frac{1}{2}\| \rho- \rho^T\|_1
\end{equation}
where the minimization is taken over all quantum states $\tau$ and $\| A \|_1=\tr(\sqrt{AA^\dagger})$ denotes the trace norm.
We now consider the imaginarity of pure states $\rho=\ket{\psi}\bra{\psi}$ 
\begin{equation}\label{eq:imag_rho}
\mathcal{I}(\ket{\psi})\equiv R(\ket{\psi})^2=1-\vert\braket{\psi}{\psi^*}\vert^2\,
\end{equation}
where $\mathcal{I}=0$ only for real states and $0\le\mathcal{I}\le1$. Here, $\ket{\psi^\ast}$ indicates the complex conjugate of $\ket{\psi}$ and we used $\frac{1}{2}\left\| \ket{\psi}\bra{\psi}- \ket{\phi}\bra{\phi}\right\|_1=\sqrt{1-\vert\braket{\psi}{\phi}\vert^2}$.

When given a polynomial number of copies of a state $\ket{\psi}$ which is promised to be either a non-imaginary state or a state with high imaginarity, can we tell those two cases apart?
The answer to this question is negative, i.e. testing the
imaginarity of states requires in general exponentially
many copies.

\begin{theorem}[Imaginarity cannot be efficiently tested for states (\SM~\ref{sec:proof_imag})]\label{thm:no_test_imag}
Any tester $\mathcal{A}_{\mathcal{I}}$ for imaginarity according to Def.~\ref{def:prop-tester} for an \revC{arbitrary} $n$-qubit state $\ket{\psi}$ requires $t=\Omega(2^{n/2})$ copies of $\ket{\psi}$ for $\delta<1-n^22^{-n/2}$ and any $\beta<\delta$. 
\end{theorem}
The  proof idea is that Haar random states are highly imaginary, while being indistinguishable from an ensemble of real states. 

For the real states, we use the recently introduced $K$-subset phase states~\cite{bouland2022quantum}
\begin{equation}\label{eq:phaseKstates}
   \ket{\psi_{S,f}}=\frac{1}{\sqrt{K}}\sum_{x\in S}(-1)^{f(x)}\ket{x} 
\end{equation}
where $S$ is a set of binary strings $\{0,1\}^{n}$ with $K=\vert S\vert$ elements and  $f:\{0,1 \}^n \to \{ 0, 1\}$ a binary phase function. The ensemble of~\eqref{eq:phaseKstates} over all $f$ has been shown to be pseudorandom~\cite{brakerski2019pseudo,bouland2022quantum}. 
The ensemble of Haar random states and the ensemble of phase states of all binary phase functions with $K=\omega(\log(n))$ are indistinguishable with polynomial copies~\cite{bouland2022quantum}, which we apply to prove Theorem~\ref{thm:no_test_imag}.

\revA{Thus, PRSs can assume any value of imaginarity as it cannot be tested efficiently, which is evidenced by the existence of both real~\cite{brakerski2019pseudo} and imaginary PRSs~\cite{ji2018pseudorandom}.}
\revC{Note that alternative constructions of real-valued PRSs exist, which do not require negative amplitudes~\cite{giurgica2023pseudorandomness,jeronimo2024pseudorandom}.}

While testing imaginarity is inefficient in general, the scaling is indeed better than tomography as we show with an explicit algorithm in \SM~\ref{sec:meas_imag_ineff}.

\revB{Also note that when we have access to both $\ket{\psi}$ and its complex conjugate $\ket{\psi^\ast}$, we can simply use the SWAP test to efficiently evaluate~\eqref{eq:imag_rho} using $O(1)$ samples. Thus, access to the complex conjugate gives an exponential advantage in imaginarity testing.
\begin{corollary}[Exponential separation in testing imaginarity with conjugate states]\label{cor:seperation}
    Testing the imaginarity of \revC{arbitrary} state $\ket{\psi}$ with access to its complex conjugate $\ket{\psi^\ast}$ requires $O(1)$ copies, while without $\ket{\psi^\ast}$ one requires $t=\Omega(2^{n/2})$ copies.  
\end{corollary}}

\revC{Now, let us assume that we are not given arbitrary states, but states drawn from a restricted class of states. In this case, efficient algorithms to test imaginarity can exist.} In particular, we can efficiently test the imaginarity of pure stabilizer states, i.e.~states that are generated by compositions of the CNOT, Hadamard, and S gates applied to computational basis states~\cite{gottesman1998heisenberg}: 
\begin{claim}[Efficient measurement of imaginarity for stabilizer states (\SM~\ref{sec:stabMeasimag})]\label{thm:meas_imag_stab}
For an $n$-qubit stabilizer state $\ket{\chi}\in\mathrm{STAB}$, imaginarity is given by
\begin{equation}
\mathcal{I}_\mathrm{STAB}(\ket{\chi})=1-2^{-n}\sum_{\sigma\in\mathcal{P}}\vert\bra{\chi}\sigma\ket{\chi^*}\vert^2\vert\bra{\chi}\sigma\ket{\chi}\vert^2
\end{equation}
where $\mathcal{P}$ is the set of all $n$-qubit Pauli strings with phase $+1$.
$\mathcal{I}_\mathrm{STAB}(\ket{\chi})$ can be efficiently measured within additive precision $\delta$ with a failure probability $\nu$ using at most $t=O(\delta^{-2}\log(1/\nu))$ copies.
\end{claim}

\prlsection{Imaginarity of unitaries}%
A channel is real if and only if its associated Choi-state is a real state~\cite{hickey2018quantifying}. Using this fact, we define imaginarity for an $n$-qubit unitary $U$ via its Choi-state $U\otimes I \ket{\Phi}$ where $\ket{\Phi}=2^{-n/2}\sum_{k=0}^{2^n-1}\ket{k}\ket{k}$ is the maximally entangled state
\begin{equation}
\mathcal{I}_\text{p}(U) \equiv  \mathcal{I}(U\otimes I \ket{\Phi})= 1-4^{-n}\left\vert\tr(U^\dagger U^\ast)\right\vert^2\,.
\end{equation}
In contrast to states,  imaginarity of unitaries can be efficiently tested:
\begin{claim}[Imaginarity of unitaries can be measured efficiently (\SM~\ref{sec:meas_imag_U} or~\cite{huang2021demonstrating})]\label{clm:meas_unitary}
Measuring $\mathcal{I}_\text{p}(U)$ within additive precision $\delta$ with a failure probability $\nu$ requires at most $t=O(\delta^{-2}\log(1/\nu))$ queries to $U$.
\end{claim}
The efficient algorithm to measure $\mathcal{I}_\text{p}$ is found in \SM~\ref{sec:meas_imag_U} or~\cite{huang2021demonstrating}. 
\revA{An earlier quantum interactive proof to distinguish Haar random complex and orthogonal unitaries is given in Ref.~\cite{aharonov2021quantum}.}
This surprising difference between unitaries and states can be intuitively understood with the following argument: Efficient measurement of imaginarity for a state $\ket{\psi}$ requires the complex conjugate $\ket{\psi^\ast}$ which cannot be efficiently prepared when having only access to copies of $\ket{\psi}$ (see also Corollary~\ref{cor:efficientstar}). In contrast, for unitary $U$ one can use the ricochet property $(U\otimes I_n)\ket{\Phi}=(I_n\otimes (U^\ast)^{\dagger})\ket{\Phi}$ to gain indirect access to the complex conjugate $U^\ast$.

\begin{theorem}[\SM~\ref{sec:imag_PRU}]\label{thm:imag_unitary}
    PRUs require imaginarity $\mathcal{I}_\text{p}=1-\mathrm{negl}(n)$. 
\end{theorem}
Theorem~\ref{thm:imag_unitary} is proven using Claim~\ref{clm:meas_unitary}. In particular, one can efficiently distinguish Haar random unitaries from low-imaginarity unitaries, which excludes them from being PRUs.

\begin{corollary}\label{cor:real_U}
    Real unitaries are not pseudorandom. 
\end{corollary}
Corollary~\ref{cor:real_U} follows from Theorem~\ref{thm:imag_unitary} by considering the special case when $U$ satisfies $\mathcal{I}_\text{p}(U)=0$.
\revC{In particular, pseudorandom permutations, which were proposed as candidates for PRUs in Ref.~\cite{ji2018pseudorandom}, cannot be PRUs due to Corollary~\ref{cor:real_U}. 
Additionally, we note that Corollary~\ref{cor:real_U} does not apply to pseudorandom isometries, which can be implemented using only real operations~\cite{ananth2024pseudorandom,metger2024simple}. Moreover, we note that by assuming a weaker notion of security where the unitary is applied on non-entangled input states, one can achieve somewhat secure PRUs that are real-valued~\cite{brakerski2024real}.
}

\prlsection{Coherence of states}%
Coherence measures the ``amount'' of off-diagonal coefficients of a state $\rho$. One measure of coherence is the relative entropy of coherence, which is defined as~\cite{baumgratz2014quantifying,streltsov2017colloquium}
\begin{align*}
    \mathcal{C}(\rho)&=\text{tr}(\rho\ln(\rho))-\text{tr}(\rho_\text{d}\ln(\rho_\text{d}))\\
    &=-\sum_k{\vert c_k\vert^2}\ln(\vert c_k\vert^2)\,,\numberthis
\end{align*}
where $\rho_\text{d} = \sum_k \bra k \rho  \ket k \ket k \! \bra k
$ is obtained by $\rho$ by keeping only the diagonal coefficients and setting all other entries to zero, and $c_k$ are the amplitudes of $\ket{\psi}=\sum_k c_k\ket{k}$ in the computational basis. Coherence is minimal for computational basis states $\mathcal{C}(\ket{k})=0$ and maximal for $\ket{+}^{\otimes n}$, $\ket{+}=\frac{1}{\sqrt{2}}(\ket{0}+\ket{1})$ with $\mathcal{C}(\ket{+}^{\otimes n})=n\ln(2)$.
We now test for the property of coherence:
\begin{theorem}[Lower bound on coherence testing for states (\SM~\ref{sec:coherenceTesting})]\label{eq:no_test_coh}
Any tester $\mathcal{A}_{\mathcal{C}}$ for relative entropy of coherence $\mathcal{C}$ according to Def.~\ref{def:prop-tester} for an \revC{arbitrary} $n$-qubit state $\ket{\psi}$ requires $t=\Omega(2^{\beta/2})$ for $\delta<n\ln(2)-1$ and any $\beta<\delta$.
\end{theorem}
The proof idea for Theorem~\ref{eq:no_test_coh} is similar to the one for imaginarity testing. First, we construct an ensemble of low coherence states via the ensemble of $K$-subset phase states with $\beta=\ln(K)$. 
Next, we build the ensemble of highly coherent states $\ket{\psi}\in \mathcal{E}_{\mathcal{C}\ge \delta}$ with coherence $\mathcal{C}(\ket{\psi})\ge \delta$ which cannot be efficiently distinguished from Haar random states~\cite{singh2016average}. 
We then use the closeness between Haar random states and subset phase states~\cite{bouland2022quantum} to prove Theorem~\ref{eq:no_test_coh}.

\begin{claim}[\SM~\ref{sec:coherenceTesting}]\label{clm:pseudo_coh}
PRSs require $\mathcal{C}=\omega(\log(n))$ relative entropy of coherence.
\end{claim}
To prove Claim~\ref{clm:pseudo_coh}, we provide an efficient measurement scheme for the coherence measure based on the Hilbert-Schmidt norm $\sum_k\vert c_k\vert^4$~\cite{baumgratz2014quantifying}. Together with Jensen's inequality, this implies that one can efficiently distinguish Haar random states from states with $\mathcal{C}=O(\log(n))$.

\begin{corollary}\label{cor:sparse_state}
States with polynomial support on the computational basis, i.e.~$\ket{\psi}=\sum_{k\in \vert S \vert } c_k \ket{k}$ with $\vert S \vert =\mathrm{poly}(n)$, are not pseudorandom.
\end{corollary}
To show Corollary~\ref{cor:sparse_state}, assume a state  $\ket{\psi}=\sum_{k\in \vert S \vert } c_k \ket{k}$ with $\vert S \vert =O(n^c)$ with $c>0$. Then, we have for the max relative entropy $\ln(\text{card}(\ket{\psi}))=O(\log(n))$, where $\text{card}(.)$ counts the non-zero coefficients of $\ket{\psi}$ in the computational basis. The state cannot be a PRS due to $\ln(\text{card}(\ket{\psi}))\ge \mathcal{C}(\ket{\psi})$ and Claim~\ref{clm:pseudo_coh}.

\prlsection{Coherence power of unitaries}%
Similar to the relative entropy of coherence for states, one can define the coherence power of unitaries~\cite{zanardi2017coherence}. 
We introduce the relative entropy of coherence power as
\begin{equation}
\mathcal{C}_\text{p}(U)=-\sum_{k,\ell} 2^{-n}\vert \bra{k}U\ket{\ell}\vert ^2\ln(\vert \bra{k}U\ket{\ell} \vert ^2)
\end{equation}
with $0\le \mathcal{C}_\text{p} \le n\ln(2)$. It measures the amount of off-diagonal coefficients of unitaries, where diagonal unitaries have $\mathcal{C}_\text{p}(U_\text{diag})=0$ and for the Hadamard gate $H$ we have $\mathcal{C}_\text{p}(H^{\otimes n})=n\ln(2)$.

\begin{theorem}[PRUs need coherence (\SM~\ref{sec:coh_unitary})]\label{thm:pseudo_coh_unitary}
PRUs require $\mathcal{C}_\text{p}=\omega(\log(n))$ relative entropy of coherence power.
\end{theorem}
The proof idea is to efficiently measure $\sum_{kl}2^{-n}\vert \bra{k}U\ket{\ell} \vert ^4$ (see \SM~\ref{sec:coh_unitary} or~\cite{zanardi2017coherence}). Using Jensen's inequality, this implies that one can efficiently distinguish unitaries with $\mathcal{C}_\text{p}=O(\log(n))$ from Haar random unitaries.

\begin{corollary}\label{cor:sparse_U}
Sparse unitaries, i.e. unitaries where rows or columns have at most $\text{poly}(n)$ non-zero entries, cannot be pseudorandom. 
\end{corollary}
To prove Corollary~\ref{cor:sparse_U}, we assume sparse unitaries $U_\text{sp}$ with $O(n^{c})$ non-zero entries per column and $c>0$. Thus, the max relative entropy is given by $\ln(2^{-n}\text{card}(U_\text{sp}))=O(\log(n))$ where $\text{card}(\cdot)$ counts the number of non-zero coefficients. From the inequalities between R\'enyi entropies, we have $\ln(2^{-n}\text{card}(U_\text{sp}))\ge \mathcal{C}_\text{p}(U_\text{sp})$, which according to Theorem~\ref{thm:pseudo_coh_unitary} precludes PRUs.

\section{Pseudoresource}
Quantum resources such as entanglement or coherence are essential ingredients to run quantum computers and devices.
We ask now whether `high'-resource state ensembles can be mimicked using another 'low´-resource ensemble? 
This concept was first introduced for the resource of entanglement as pseudoentanglement in Ref.~\cite{bouland2022quantum}.
Pseudoentanglement describes ensembles of states which are computationally indistinguishable, yet have widely different entanglement. %

We generalise this concept to \emph{pseudoresources} for arbitrary resource measures $Q(\ket{\psi})$ with $0\le Q\le Q_\text{max}(n)$, where $Q_\text{max}(n)$ is the maximal value of $Q$. For formal definition see Def.~\ref{def:pseudoresource}
A pseudoresource ensemble consists of two state ensembles: The first ensemble $\ket{\psi_k}$ has a high amount of the resource $Q(\ket{\psi_k})=f(n)$. The second ensemble $\ket{\phi_k}$ has a lower amount of the resource $Q(\ket{\psi_k})=g(n)$ while being computationally indistinguishable from the first ensemble, where $f(n)\ge g(n)$ are two functions. Thus,  a pseudoresource ensemble can masquerade as having $f(n)$ of the resource, while actually containing only $g(n)$. The difference in resource between the two ensembles is called the pseudoresource gap $f(n)$ vs $g(n)$. We summarize results for various pseudoresources in Table~\ref{tab:pseudoresource}.

Next, we introduce three new types of pseudoresources: Pseudocoherence, pseudoimaginarity and pseudopurity.

We find pseudocoherent ensembles with a pseuodcoherence gap of $f(n)=\Theta(n)$ vs $g(n)=\log(n)$, by tuning the rank of the subset phase states. This shows that one can mimic extensive coherence with only logarithmic coherent states. As coherence can be enhanced independently of entanglement and magic via LOCC and Clifford operations, pseudocoherence can be tuned independently of pseudoentanglement and pseudomagic.

In contrast, for purity $\text{tr}(\rho^2)$, the pseudopurity gap is exponentially small with $g(n)=1-O(2^{-n})$ vs $f(n)=1$ as any state with non-negligible noise is not pseudorandom.

For imaginarity, we find that the pseudoimaginarity gap is maximal with $f(n)=1-O(2^{-n})$ vs $g(n)=0$. In particular, subset phase states have zero imaginarity, while a class of imaginary pseudorandom states~\cite{ji2018pseudorandom} have asymptotically maximal imaginarity.

We categorize three classes of pseudoresources using the pseudoresource gap ratio $\Delta=(f(n)-g(n))/Q_\text{max}$, where $0\le \Delta\le 1$. $\Delta$ indicates the degree to which a resource can be masqueraded using pseudorandom states.
First, pseudopurity has $\Delta=O(2^{-n})$, indicating that only maximally pure states can be pseudorandom. 
Second, for pseudoentanglement, pseudomagic and pseudocoherence we find $\Delta=1-\omega(\log(n)/n)$. For this class of resources a low, but non-negligible amount of the resource is needed such that pseudorandom states can masquerade as high-resource states.
The third class, which consists of  pseudoimaginarity, any amount of imaginarity can be mimicked by pseudorandom states. Thus, this resource can be freely chosen for pseudorandom states.

\begin{table*}[htbp]\centering
\begin{tabular}{| c|c|c|c| }
    \hline    
  Resource $Q$  & $f(n)$ & $g(n)$
            & $\Delta$\\
    \hline\hline
    Purity   & $1$& $1-O(2^{-n})$ & $O(2^{-n})$      \\ \hline
    Entanglement~\cite{bouland2022quantum}   &   $\Theta(n)$  &   $\omega(\log(n))$ &   $1-\omega(\log(n)/n)$ \\
        \hline
    Magic~\cite{gu2023little}  &$\Theta(n)$  &   $\omega(\log(n))$ &   $1-\omega(\log(n)/n)$  \\ \hline
    Coherence  & $\Theta(n)$  &   $\omega(\log(n))$ &   $1-\omega(\log(n)/n)$       \\
    \hline
Imaginarity   &    $1-O(2^{-n})$& $0$ & $1-O(2^{-n})$ \\
    \hline
\end{tabular}
\caption{Maximal pseudoresource gap $f(n)$ vs $g(n)$ of pure states (except for purity) for different resources $Q$. We also show the pseudoresource gap ratio $\Delta=(f(n)-g(n))/Q_\text{max}$ with respect to the maximal value of the resource $Q_\text{max}$. 
}
\label{tab:pseudoresource}
\end{table*}

\prlsection{Separation between PRSS and PRUs}
\revA{Any PRU is also a PRSS~\cite{lu2023quantum,ananth2024pseudorandom}. However, the reverse question whether any PRSS is also a PRU has been left as an open question~\cite{lu2023quantum,ananth2024pseudorandom}, which we will proceed to answer negatively:
\begin{theorem}[Separation between PRSS and PRU]\label{thm:PRSS_proof}
PRSS are strictly weaker than PRSS. In particular, there are PRSS which are not PRUs.
\end{theorem}
In Ref.~\cite{lu2023quantum} (Def. 5.4), a class of real-valued PRSS based on Kac's walk is proposed. It uses the permutation $\sigma \in S_{2^n}$ over $n$ qubits, 
and real-valued rotations with discretized rotation angles. This PRSS is not a PRU due to a lack of imaginarity, which follows directly from Thm.~\ref{thm:imag_unitary_sup}.}

\prlsection{Transformation limits}
Universal quantum computation with qubits usually requires complex-valued operations and states. 
However, one can also realize the same universal quantum computations using $n+1$ rebits, i.e.~qubits with only real coefficients which support only real states and real unitaries~\cite{rudolph20022,koh2018quantum,mckague2013power,delfosse2015wigner}. For rebits, the role of imaginarity is mimicked by adding an ancilla rebit which acts as a flag that records the real and imaginary parts of the amplitudes of the quantum states separately. In particular, while a complex $n$-qubit state with coefficients $a_k,b_k\in\mathbb{R}$ is written as 
\begin{equation}\label{eq:qubit}
    \ket{\psi}=\sum_{k=0}^{2^n-1}(a_k+ib_k)\ket{k}\,
\end{equation}
the corresponding $(n+1)$-rebit state is given by 
\begin{equation}\label{eq:rebit}
    \ket{\psi_\text{R}}=\sum_{k=0}^{2^n-1}( a_k\ket{k}\ket{R}+b_k\ket{k}\ket{I})
\end{equation}
where $\ket{R}$, $\ket{I}$ are the computational flag states for real and imaginary value respectively. 
The rebit representation can naturally perform additional non-linear operations not commonly associated with the usual qubit quantum computing model, and thus it would be highly interesting to transform between the qubit and rebit representations~\cite{koh2018quantum}. 
However, as we show below, qubit-to-rebit transformations are inefficient, in contrast to the reverse process, for which we give an efficient protocol.
\begin{corollary}\label{cor:rebit}
There is no algorithm $\mathcal{T}(\ket{\psi}^{\otimes t}) \rightarrow \ket{\psi_\text{R}}$
 using $t=\mathrm{poly}(n)$ copies of $n$-qubit state~\eqref{eq:qubit} 
 that returns the $(n+1)$-rebit state~\eqref{eq:rebit} 
 with non-negligible probability. 
\end{corollary}
Corollary~\ref{cor:rebit} can be easily seen from the following: In the rebit representation $\ket{\psi_\text{R}}$, we can efficiently measure the corresponding imaginarity of the qubit representation $\ket{\psi}$ via the purity of the flag rebit~(\SM~\ref{sec:rebits})
\begin{equation}\label{eq:rebit_imag}
\mathcal{I}(\ket{\psi})\equiv 2(1-\tr(\tr_{1:n}(\ketbra{\psi_\text{R}}{\psi_\text{R}})^2))\,,
\end{equation}
where $\tr_{1:n}(\cdot)$ is the partial trace over all rebits except the (${n+1}$)-th rebit. 
\revC{Now, suppose for the sake of contradiction that $\mathcal{T}$ is efficient. Then, we could efficiently transform $\ket{\psi}$ into $\ket{\psi_\text{R}}$ and efficiently measure the imaginarity of $\ket{\psi}$ via $\ket{\psi_\text{R}}$. However, the efficient measurement of imaginarity contradicts Theorem~\ref{thm:no_test_imag}. Therefore, $\mathcal{T}$ must be inefficient.}

\begin{claim}\label{cor:rebit_inv}
There is an efficient algorithm $\mathcal{T}_\text{R}(\ket{\psi_\text{R}}^{\otimes {t}})\rightarrow \ket{\psi}$ $t$ copies of $(n+1)$-rebit state~\eqref{eq:rebit}
that returns $n$-qubit state~\eqref{eq:qubit}
\revC{with exponentially small failure probability $O(2^{-t})$}. 
\end{claim}
Algorithm $\mathcal{T}_\text{R}$ uses the operator $\Pi=I_n\otimes \frac{1}{\sqrt{2}}(\bra{R}+i\bra{I})$ via $\ket{\psi}=\sqrt{2}\Pi\ket{\psi_\text{R}}$, where $I_n$ is the $n$-qubit identity operator~\cite{koh2018quantum}. $\Pi$ is implemented efficiently by measuring the flag rebit in the $\{\frac{1}{\sqrt{2}}(\ket{R}+i\ket{I}),\frac{1}{\sqrt{2}}(\ket{R}-i\ket{I})\}$ basis and post-selecting $\frac{1}{\sqrt{2}}(\ket{R}+i\ket{I})$ with success probability $\left\|\Pi\ket{\psi_\text{R}}\right\|=\frac{1}{2}$. \revC{By repeating the above algorithm with $t$ copies, one succeeds at least once with failure probability $2^{-t}$.}

Finally, we give a short proof that the complex conjugation of states is inefficient. This was previously proven using more involved techniques~\cite{yang2014certifying,miyazaki2019complex}:
\begin{corollary}\label{cor:efficientstar}
There is no efficient quantum algorithm with access to $\mathrm{poly}(n)$ copies of $\ket{\psi}$ that returns the complex conjugate $\ket{\psi^*}$ with non-negligible probability.
\end{corollary}
\revC{One can show this as follows: When given $t$ copies of $\ket{\psi}$ and $\ket{\psi^*}$, one can efficiently test imaginarity $\mathcal{I}(\ket{\psi})$ up to additive accuracy $O(1/\sqrt{t})$ via the SWAP test. This scheme can be used as a property tester for imaginarity, i.e. one can test with high probability whether any given state has zero imaginarity or high imaginarity. Thus, efficient access to the complex conjugate state contradicts Theorem~\ref{thm:no_test_imag}, implying that $\ket{\psi^*}$ cannot be efficiently prepared from $\ket{\psi}$. }

\prlsection{Conclusion}%
We have studied the limits of testing purity, imaginarity and coherence. This, in turn, places fundamental limits on the properties of PRSs and PRUs as a function of qubit number $n$. Our work builds upon the fundamental and widely acknowledged principle that a pseudorandom object possesses all the efficiently verifiable properties inherent in its random counterpart~\cite{katz2020introduction}.

PRSs and PRUs are not robust to noise, and can tolerate depolarizing noise with at most negligible probability. Thus, it is not possible to generate PRSs on noisy intermediate-scale quantum computers~\cite{bharti2021noisy}. Further, early fault-tolerant quantum computers~\cite{babbush2021focus,suzuki2022quantum}, i.e. quantum computers with limited error correction where not all errors are exponentially suppressed, are expected to be unable to generate PRSs. 

We further show that PRSs and PRUs must have a $\mathcal{C}=\omega(\log(n))$ relative entropy of coherence (power), which implies that PRUs are not sparse, requiring  $\omega(\log(n))$ coherent operations such as the Hadamard gates to be generated.

Does the description of quantum mechanics require complex numbers, or does it suffice to only use real numbers to achieve arbitrary tasks? For quantum communication tasks, it has been shown that complex numbers are indeed necessary~\cite{wu2021resource,kondra2022real,chen2022ruling}.
We now show that quantum randomness also requires complex numbers, however with a subtle catch.
In particular, we show that PRUs must not be real and require $\mathcal{I}_\text{p}=1-\text{negl}(n)$ imaginarity.  
This implies that a real-valued description of quantum information is unable to generate quantum randomness for unitaries, and complex numbers are indeed necessary. 
In contrast, PRSs (and PRSSs) have no restrictions on imaginarity. Thus, randomness on the level of quantum states does not need complex numbers, and a real-valued quantum theory suffices to generate states which look completely random. This reveals a fundamental difference of the role of complex numbers for operations and states. Given recent efficient constructions of real PRSs~\cite{brakerski2019pseudo}, some-what secure real PRUs~\cite{brakerski2024real}, and complex PRUs~\cite{schuster2024randomunitariesextremelylow,ma2024construct}, the relationship of randomness and complex numbers can be directly probed in experiment.

We introduce pseudoresource as the ability of a pseudorandom state ensemble with a low amount of a given resource to masquerade as a high-resource ensemble.
We observe that pseudoresource behaves fundamentally different depending on the resource: For pseudocoherence (as well as pseudoentanglement and pseudomagic), we find a large, but non-maximal pseudoresource gap between low and high coherence ensembles. 
For pseudoimaginarity, the gap is maximal, while for pseudopurity the gap is exponentially small. 
Thus, not all resources can be effectively masqueraded using low-resource states. It would be interesting to draw a fundamental connection between resource theories and the pseudoresource gap. \revC{Further, it would be interesting to relate pseudoresources to state preparation complexity~\cite{grewal2024pseudoentanglement}.}

Various notions of pseudorandomness exist, and their relative relationships is not well understood. Here, we establish two new results: First, we show that creating PRUs is more demanding than generating PRSs, which was conjectured in Ref.~\cite{ji2018pseudorandom}.
In particular, while real and diagonal unitaries applied to a product state can create the PRSs of~\eqref{eq:phaseKstates}~\cite{ji2018pseudorandom}, we show that these unitaries cannot be pseudorandom themselves, providing a separation in the generation complexity of PRSs and PRUs.
\revA{Second, we show that PRSS are not always PRUs, proving that there is indeed a clear gap between the two notions of pseudorandomness which was alluded in Refs.~\cite{lu2023quantum,ananth2024pseudorandom}.}

\revB{We show that access to the complex conjugate of a state gives an exponential advantage in imaginarity testing. Thus, the ability to conjugate (which in general is exponentially difficult) poses a powerful resource in property testing, which similarly has been shown for learning problems~\cite{wu2023quantum,king2024exponential}. }

Our work also sheds light on the relationship between rebit and qubit models of quantum computing. In particular, we show that one can efficiently transform from rebit to a qubit model of quantum computing, while the inverse transformation is not efficient.

Finally, we note that for pure states and unitaries, one can efficiently test quantum properties such as entanglement~\cite{ekert2002direct,bendersky2009general,harrow2013testing}, magic~\cite{gross2021schur,haug2022scalable}, coherence~\cite{zanardi2017coherence}, purity~\cite{barenco1997stabilization} and imaginarity~\cite{huang2021demonstrating}, with the sole exception being the imaginarity of states (see \SM~\ref{sec:comparison}). This curious outlier highlights the special character of the resource theory of imaginarity, which warrants further studies.

\medskip

\begin{acknowledgments}
D.E.K.\ acknowledges funding support from the National Research Foundation, Singapore, and the Agency for Science, Technology and Research (A*STAR), Singapore, under its Quantum Engineering Programme (NRF2021-QEP2-02-P03); A*STAR C230917003; and A*STAR under the Central Research Fund (CRF) Award for Use-Inspired Basic Research (UIBR) and the Quantum Innovation Centre (Q.InC) Strategic Research and Translational Thrust. We thank Nikhil Bansal, Naresh Goud Boddu, Atul Singh Arora and Rahul Jain for various discussions. We also extend our thanks to the anonymous referee for pointing out the candidate construction for PRUs from Ref.~\cite{ji2018pseudorandom}, which is ruled out as a result of our findings.
\end{acknowledgments}

\bibliographystyle{quantum}
\bibliography{imaginary}

\begin{thebibliography}{100}

\bibitem{ji2018pseudorandom}
Zhengfeng Ji, Yi-Kai Liu, and Fang Song.
\newblock ``Pseudorandom quantum states''.
\newblock In Annual International Cryptology Conference.
\newblock \href{https://dx.doi.org/10.1007/978-3-319-96878-0_5}{Pages
  126--152}.
\newblock Springer~(2018).

\bibitem{emerson2003pseudo}
Joseph Emerson, Yaakov~S Weinstein, Marcos Saraceno, Seth Lloyd, and David~G
  Cory.
\newblock ``Pseudo-random unitary operators for quantum information
  processing''.
\newblock \href{https://dx.doi.org/10.1126/science.1090790}{Science {\bf 302},
  2098--2100}~(2003).

\bibitem{brandao2016efficient}
Fernando G. S.~L. Brand\~ao, Aram~W. Harrow, and Micha\l{} Horodecki.
\newblock ``Efficient quantum pseudorandomness''.
\newblock \href{https://dx.doi.org/10.1103/PhysRevLett.116.170502}{Phys. Rev.
  Lett. {\bf 116}, 170502}~(2016).

\bibitem{brandao2016local}
Fernando G S~L Brandao, Aram~W Harrow, and Micha{\l} Horodecki.
\newblock ``Local random quantum circuits are approximate polynomial-designs''.
\newblock \href{https://dx.doi.org/10.1007/s00220-016-2706-8}{Communications in
  Mathematical Physics {\bf 346}, 397--434}~(2016).

\bibitem{nakata2017efficient}
Yoshifumi Nakata, Christoph Hirche, Masato Koashi, and Andreas Winter.
\newblock ``Efficient quantum pseudorandomness with nearly time-independent
  {H}amiltonian dynamics''.
\newblock \href{https://dx.doi.org/10.1103/PhysRevX.7.021006}{Phys. Rev. X {\bf
  7}, 021006}~(2017).

\bibitem{scarani2010guaranteed}
Valerio Scarani.
\newblock ``Guaranteed randomness''.
\newblock \href{https://dx.doi.org/10.1038/464988a}{Nature {\bf 464},
  988--989}~(2010).

\bibitem{law2014quantum}
Yun~Zhi Law, Jean-Daniel Bancal, Valerio Scarani, et~al.
\newblock ``Quantum randomness extraction for various levels of
  characterization of the devices''.
\newblock \href{https://dx.doi.org/10.1088/1751-8113/47/42/424028}{Journal of
  Physics A: Mathematical and Theoretical {\bf 47}, 424028}~(2014).

\bibitem{herrero2017quantum}
Miguel Herrero-Collantes and Juan~Carlos Garcia-Escartin.
\newblock ``Quantum random number generators''.
\newblock \href{https://dx.doi.org/10.1038/npjqi.2016.21}{Reviews of Modern
  Physics {\bf 89}, 015004}~(2017).

\bibitem{acin2016certified}
Antonio Ac{\'\i}n and Lluis Masanes.
\newblock ``Certified randomness in quantum physics''.
\newblock \href{https://dx.doi.org/10.1038/nature20119}{Nature {\bf 540},
  213--219}~(2016).

\bibitem{bera2017randomness}
Manabendra~Nath Bera, Antonio Ac{\'\i}n, Marek Ku{\'s}, Morgan~W Mitchell, and
  Maciej Lewenstein.
\newblock ``Randomness in quantum mechanics: philosophy, physics and
  technology''.
\newblock \href{https://dx.doi.org/10.1088/1361-6633/aa8731}{Reports on
  Progress in Physics {\bf 80}, 124001}~(2017).

\bibitem{calude2008quantum}
Cristian~S Calude and Karl Svozil.
\newblock ``Quantum randomness and value indefiniteness''.
\newblock \href{https://dx.doi.org/10.1166/asl.2008.016}{Advanced Science
  Letters {\bf 1}, 165--168}~(2008).

\bibitem{schuster2024randomunitariesextremelylow}
Thomas Schuster, Jonas Haferkamp, and Hsin-Yuan Huang.
\newblock ``Random unitaries in extremely low depth''.
\newblock \href{https://dx.doi.org/10.48550/arXiv.2407.07754}{{a}rXiv preprint
  {a}rXiv:2407.07754}~(2024).

\bibitem{ma2024construct}
Fermi Ma and Hsin-Yuan Huang.
\newblock ``How to construct random unitaries''.
\newblock \href{https://dx.doi.org/10.48550/arXiv.2410.10116}{{a}rXiv preprint
  {a}rXiv:2410.10116}~(2024).

\bibitem{lu2023quantum}
Chuhan Lu, Minglong Qin, Fang Song, Penghui Yao, and Mingnan Zhao.
\newblock ``Quantum pseudorandom scramblers''.
\newblock \href{https://dx.doi.org/10.48550/arXiv.2309.08941}{{a}rXiv preprint
  {a}rXiv:2309.08941}~(2023).

\bibitem{ananth2024pseudorandom}
Prabhanjan Ananth, Aditya Gulati, Fatih Kaleoglu, and Yao-Ting Lin.
\newblock ``Pseudorandom isometries''.
\newblock In Annual International Conference on the Theory and Applications of
  Cryptographic Techniques.
\newblock \href{https://dx.doi.org/10.1007/978-3-031-58737-5_9}{Pages
  226--254}.
\newblock Springer~(2024).

\bibitem{brakerski2019pseudo}
Zvika Brakerski and Omri Shmueli.
\newblock ``({P}seudo) random quantum states with binary phase''.
\newblock In Theory of Cryptography Conference.
\newblock \href{https://dx.doi.org/10.1007/978-3-030-36030-6_10}{Pages
  229--250}.
\newblock Springer~(2019).

\bibitem{bouland2022quantum}
Scott Aaronson, Adam Bouland, Bill Fefferman, Soumik Ghosh, Umesh Vazirani,
  Chenyi Zhang, and Zixin Zhou.
\newblock ``{Quantum Pseudoentanglement}''.
\newblock In Venkatesan Guruswami, editor, 15th Innovations in Theoretical
  Computer Science Conference (ITCS 2024).
\newblock \href{https://dx.doi.org/10.4230/LIPIcs.ITCS.2024.2}{Volume 287 of
  Leibniz International Proceedings in Informatics (LIPIcs), pages 2:1--2:21}.
\newblock Dagstuhl, Germany~(2024). Schloss Dagstuhl -- Leibniz-Zentrum f{\"u}r
  Informatik.

\bibitem{kretschmer2021quantum}
William Kretschmer.
\newblock ``{Quantum Pseudorandomness and Classical Complexity}''.
\newblock In Min-Hsiu Hsieh, editor, 16th Conference on the Theory of Quantum
  Computation, Communication and Cryptography (TQC 2021).
\newblock \href{https://dx.doi.org/10.4230/LIPIcs.TQC.2021.2}{Volume 197 of
  Leibniz International Proceedings in Informatics (LIPIcs), pages 2:1--2:20}.
\newblock Dagstuhl, Germany~(2021). Schloss Dagstuhl -- Leibniz-Zentrum f{\"u}r
  Informatik.

\bibitem{ananth2022cryptography}
Prabhanjan Ananth, Luowen Qian, and Henry Yuen.
\newblock ``Cryptography from pseudorandom quantum states''.
\newblock In Advances in Cryptology--CRYPTO 2022: 42nd Annual International
  Cryptology Conference, CRYPTO 2022, Santa Barbara, CA, USA, August 15--18,
  2022, Proceedings, Part I.
\newblock \href{https://dx.doi.org/10.1007/978-3-031-15802-5_8}{Pages
  208--236}.
\newblock Springer~(2022).

\bibitem{ananth2022pseudorandom}
Prabhanjan Ananth, Aditya Gulati, Luowen Qian, and Henry Yuen.
\newblock ``Pseudorandom (function-like) quantum state generators: New
  definitions and applications''.
\newblock In Theory of Cryptography Conference.
\newblock \href{https://dx.doi.org/10.1007/978-3-031-22318-1_9}{Pages
  237--265}.
\newblock Springer~(2022).

\bibitem{morimae2022quantum}
Tomoyuki Morimae and Takashi Yamakawa.
\newblock ``Quantum commitments and signatures without one-way functions''.
\newblock In Advances in Cryptology--CRYPTO 2022: 42nd Annual International
  Cryptology Conference, CRYPTO 2022, Santa Barbara, CA, USA, August 15--18,
  2022, Proceedings, Part I.
\newblock \href{https://dx.doi.org/10.1007/978-3-031-15802-5_10}{Pages
  269--295}.
\newblock Springer~(2022).

\bibitem{ananth2023pseudorandom}
Prabhanjan Ananth, Yao-Ting Lin, and Henry Yuen.
\newblock ``{Pseudorandom Strings from Pseudorandom Quantum States}''.
\newblock In Venkatesan Guruswami, editor, 15th Innovations in Theoretical
  Computer Science Conference (ITCS 2024).
\newblock \href{https://dx.doi.org/10.4230/LIPIcs.ITCS.2024.6}{Volume 287 of
  Leibniz International Proceedings in Informatics (LIPIcs), pages 6:1--6:22}.
\newblock Dagstuhl, Germany~(2024). Schloss Dagstuhl -- Leibniz-Zentrum f{\"u}r
  Informatik.

\bibitem{grewal2023improved}
Sabee Grewal, Vishnu Iyer, William Kretschmer, and Daniel Liang.
\newblock ``Improved stabilizer estimation via {B}ell difference sampling''.
\newblock In Proceedings of the 56th Annual ACM Symposium on Theory of
  Computing.
\newblock \href{https://dx.doi.org/10.1145/3618260.3649738}{Page 1352–1363}.
\newblock STOC 2024New York, NY, USA~(2024). Association for Computing
  Machinery.

\bibitem{wu2021resource}
Kang-Da Wu, Tulja~Varun Kondra, Swapan Rana, Carlo~Maria Scandolo, Guo-Yong
  Xiang, Chuan-Feng Li, Guang-Can Guo, and Alexander Streltsov.
\newblock ``Resource theory of imaginarity: Quantification and state
  conversion''.
\newblock \href{https://dx.doi.org/10.1103/PhysRevA.103.032401}{Phys. Rev. A
  {\bf 103}, 032401}~(2021).

\bibitem{kondra2022real}
Tulja~Varun Kondra, Chandan Datta, and Alexander Streltsov.
\newblock ``Real quantum operations and state transformations''.
\newblock \href{https://dx.doi.org/10.1088/1367-2630/acf9c4}{New Journal of
  Physics {\bf 25}, 093043}~(2023).

\bibitem{chen2022ruling}
Ming-Cheng Chen, Can Wang, Feng-Ming Liu, Jian-Wen Wang, Chong Ying, Zhong-Xia
  Shang, Yulin Wu, M.~Gong, H.~Deng, F.-T. Liang, Qiang Zhang, Cheng-Zhi Peng,
  Xiaobo Zhu, Ad\'an Cabello, Chao-Yang Lu, and Jian-Wei Pan.
\newblock ``Ruling out real-valued standard formalism of quantum theory''.
\newblock \href{https://dx.doi.org/10.1103/PhysRevLett.128.040403}{Phys. Rev.
  Lett. {\bf 128}, 040403}~(2022).

\bibitem{renou2021quantum}
Marc-Olivier Renou, David Trillo, Mirjam Weilenmann, Thinh~P Le, Armin
  Tavakoli, Nicolas Gisin, Antonio Ac{\'\i}n, and Miguel Navascu{\'e}s.
\newblock ``Quantum theory based on real numbers can be experimentally
  falsified''.
\newblock \href{https://dx.doi.org/10.1038/s41586-021-04160-4}{Nature {\bf
  600}, 625--629}~(2021).

\bibitem{wu2021operational}
Kang-Da Wu, Tulja~Varun Kondra, Swapan Rana, Carlo~Maria Scandolo, Guo-Yong
  Xiang, Chuan-Feng Li, Guang-Can Guo, and Alexander Streltsov.
\newblock ``Operational resource theory of imaginarity''.
\newblock \href{https://dx.doi.org/10.1103/PhysRevLett.126.090401}{Phys. Rev.
  Lett. {\bf 126}, 090401}~(2021).

\bibitem{li2022testing}
Zheng-Da Li, Ya-Li Mao, Mirjam Weilenmann, Armin Tavakoli, Hu~Chen, Lixin Feng,
  Sheng-Jun Yang, Marc-Olivier Renou, David Trillo, Thinh~P. Le, et~al.
\newblock ``Testing real quantum theory in an optical quantum network''.
\newblock \href{https://dx.doi.org/10.1103/PhysRevLett.128.040402}{Phys. Rev.
  Lett. {\bf 128}, 040402}~(2022).

\bibitem{baumgratz2014quantifying}
T.~Baumgratz, M.~Cramer, and M.~B. Plenio.
\newblock ``Quantifying coherence''.
\newblock \href{https://dx.doi.org/10.1103/PhysRevLett.113.140401}{Phys. Rev.
  Lett. {\bf 113}, 140401}~(2014).

\bibitem{streltsov2017colloquium}
Alexander Streltsov, Gerardo Adesso, and Martin~B. Plenio.
\newblock ``Colloquium: Quantum coherence as a resource''.
\newblock \href{https://dx.doi.org/10.1103/RevModPhys.89.041003}{Rev. Mod.
  Phys. {\bf 89}, 041003}~(2017).

\bibitem{bu2017cohering}
Kaifeng Bu, Asutosh Kumar, Lin Zhang, and Junde Wu.
\newblock ``Cohering power of quantum operations''.
\newblock
  \href{https://dx.doi.org/https://doi.org/10.1016/j.physleta.2017.03.022}{Physics
  Letters A {\bf 381}, 1670--1676}~(2017).

\bibitem{zanardi2017coherence}
Paolo Zanardi, Georgios Styliaris, and Lorenzo Campos~Venuti.
\newblock ``Coherence-generating power of quantum unitary maps and beyond''.
\newblock \href{https://dx.doi.org/10.1103/PhysRevA.95.052306}{Phys. Rev. A
  {\bf 95}, 052306}~(2017).

\bibitem{bu2017maximum}
Kaifeng Bu, Uttam Singh, Shao-Ming Fei, Arun~Kumar Pati, and Junde Wu.
\newblock ``Maximum relative entropy of coherence: An operational coherence
  measure''.
\newblock \href{https://dx.doi.org/10.1103/PhysRevLett.119.150405}{Phys. Rev.
  Lett. {\bf 119}, 150405}~(2017).

\bibitem{bu2017note}
Kaifeng Bu and Chunhe Xiong.
\newblock ``A note on cohering power and de-cohering power''.
\newblock \href{https://dx.doi.org/10.26421/QIC17.13-14-8}{Quantum Information
  \& Computation {\bf 17}, 1206--1220}~(2017).

\bibitem{streltsov2015measuring}
Alexander Streltsov, Uttam Singh, Himadri~Shekhar Dhar, Manabendra~Nath Bera,
  and Gerardo Adesso.
\newblock ``Measuring quantum coherence with entanglement''.
\newblock \href{https://dx.doi.org/10.1103/PhysRevLett.115.020403}{Phys. Rev.
  Lett. {\bf 115}, 020403}~(2015).

\bibitem{chitambar2016relating}
Eric Chitambar and Min-Hsiu Hsieh.
\newblock ``Relating the resource theories of entanglement and quantum
  coherence''.
\newblock \href{https://dx.doi.org/10.1103/PhysRevLett.117.020402}{Phys. Rev.
  Lett. {\bf 117}, 020402}~(2016).

\bibitem{bu2022complexity}
Kaifeng Bu, Roy~J Garcia, Arthur Jaffe, Dax~Enshan Koh, and Lu~Li.
\newblock ``Complexity of quantum circuits via sensitivity, magic, and
  coherence''.
\newblock \href{https://dx.doi.org/10.1007/s00220-024-05030-6}{Communications
  in Mathematical Physics {\bf 405}, 161}~(2024).

\bibitem{hillery2016coherence}
Mark Hillery.
\newblock ``Coherence as a resource in decision problems: The {D}eutsch-{J}ozsa
  algorithm and a variation''.
\newblock \href{https://dx.doi.org/10.1103/PhysRevA.93.012111}{Phys. Rev. A
  {\bf 93}, 012111}~(2016).

\bibitem{matera2016coherent}
J~M Matera, D~Egloff, N~Killoran, and M~B Plenio.
\newblock ``Coherent control of quantum systems as a resource theory''.
\newblock \href{https://dx.doi.org/10.1088/2058-9565/1/1/01LT01}{Quantum
  Science and Technology {\bf 1}, 01LT01}~(2016).

\bibitem{shi2017coherence}
Hai-Long Shi, Si-Yuan Liu, Xiao-Hui Wang, Wen-Li Yang, Zhan-Ying Yang, and Heng
  Fan.
\newblock ``Coherence depletion in the {G}rover quantum search algorithm''.
\newblock \href{https://dx.doi.org/10.1103/PhysRevA.95.032307}{Phys. Rev. A
  {\bf 95}, 032307}~(2017).

\bibitem{wang2021minimizing}
Guoming Wang, Dax~Enshan Koh, Peter~D. Johnson, and Yudong Cao.
\newblock ``Minimizing estimation runtime on noisy quantum computers''.
\newblock \href{https://dx.doi.org/10.1103/PRXQuantum.2.010346}{PRX Quantum
  {\bf 2}, 010346}~(2021).

\bibitem{rubinfeld1996robust}
Ronitt Rubinfeld and Madhu Sudan.
\newblock ``Robust characterizations of polynomials with applications to
  program testing''.
\newblock \href{https://dx.doi.org/10.1137/S0097539793255151}{SIAM Journal on
  Computing {\bf 25}, 252--271}~(1996).

\bibitem{goldreich1998property}
Oded Goldreich, Shari Goldwasser, and Dana Ron.
\newblock ``Property testing and its connection to learning and
  approximation''.
\newblock \href{https://dx.doi.org/10.1145/285055.285060}{Journal of the ACM
  (JACM) {\bf 45}, 653--750}~(1998).

\bibitem{buhrman2008quantum}
Harry Buhrman, Lance Fortnow, Ilan Newman, and Hein R{\"o}hrig.
\newblock ``Quantum property testing''.
\newblock \href{https://dx.doi.org/10.1137/S009753970444241}{SIAM Journal on
  Computing {\bf 37}, 1387--1400}~(2008).

\bibitem{montanaro2013survey}
Ashley Montanaro and Ronald de~Wolf.
\newblock ``A survey of quantum property testing''.
\newblock \href{https://dx.doi.org/10.4086/toc.gs.2016.007}{Pages 1--81}.
\newblock Number~7 in Graduate Surveys. Theory of Computing Library. ~(2016).

\bibitem{harrow2013testing}
Aram~W Harrow and Ashley Montanaro.
\newblock ``Testing product states, quantum {M}erlin-{A}rthur games and tensor
  optimization''.
\newblock \href{https://dx.doi.org/10.1145/2432622.2432625}{Journal of the ACM
  (JACM) {\bf 60}, 1--43}~(2013).

\bibitem{gross2021schur}
David Gross, Sepehr Nezami, and Michael Walter.
\newblock ``Schur--{W}eyl duality for the {C}lifford group with applications:
  Property testing, a robust {H}udson theorem, and de {F}inetti
  representations''.
\newblock \href{https://dx.doi.org/10.1007/s00220-021-04118-7}{Communications
  in Mathematical Physics {\bf 385}, 1325--1393}~(2021).

\bibitem{chen2023testing}
Thomas Chen, Shivam Nadimpalli, and Henry Yuen.
\newblock ``Testing and learning quantum juntas nearly optimally''.
\newblock In Proceedings of the 2023 Annual ACM-SIAM Symposium on Discrete
  Algorithms (SODA).
\newblock \href{https://dx.doi.org/10.1137/1.9781611977554.ch43}{Pages
  1163--1185}.
\newblock SIAM~(2023).

\bibitem{she2022unitary}
Adrian She and Henry Yuen.
\newblock ``{Unitary Property Testing Lower Bounds by Polynomials}''.
\newblock In Yael Tauman~Kalai, editor, 14th Innovations in Theoretical
  Computer Science Conference (ITCS 2023).
\newblock \href{https://dx.doi.org/10.4230/LIPIcs.ITCS.2023.96}{Volume 251 of
  Leibniz International Proceedings in Informatics (LIPIcs), pages
  96:1--96:17}.
\newblock Dagstuhl, Germany~(2023). Schloss Dagstuhl -- Leibniz-Zentrum f{\"u}r
  Informatik.

\bibitem{brakerski2021unitary}
Zvika Brakerski, Devika Sharma, and Guy Weissenberg.
\newblock ``Unitary subgroup testing''.
\newblock \href{https://dx.doi.org/10.48550/arXiv.2104.03591}{{a}rXiv preprint
  {a}rXiv:2104.03591}~(2021).

\bibitem{soleimanifar2022testing}
Mehdi Soleimanifar and John Wright.
\newblock ``Testing matrix product states''.
\newblock In Proceedings of the 2022 Annual ACM-SIAM Symposium on Discrete
  Algorithms (SODA).
\newblock \href{https://dx.doi.org/10.1137/1.9781611977073.68}{Pages
  1679--1701}.
\newblock SIAM~(2022).

\bibitem{bansal2024pseudorandom}
Nikhil Bansal, Wai-Keong Mok, Kishor Bharti, Dax~Enshan Koh, and Tobias Haug.
\newblock ``Pseudorandom density matrices''.
\newblock \href{https://dx.doi.org/10.1103/PRXQuantum.6.020322}{PRX Quantum
  {\bf 6}, 020322}~(2025).

\bibitem{chen2024efficient}
Chi-Fang Chen, Jordan Docter, Michelle Xu, Adam Bouland, Fernando~G.S.L.
  Brandão, and Patrick Hayden.
\newblock ``Efficient unitary designs from random sums and permutations''.
\newblock In 2024 IEEE 65th Annual Symposium on Foundations of Computer Science
  (FOCS).
\newblock \href{https://dx.doi.org/10.1109/FOCS61266.2024.00037}{Pages
  476--484}.
\newblock ~(2024).

\bibitem{metger2024pseudorandom}
Tony Metger, Alexander Poremba, Makrand Sinha, and Henry Yuen.
\newblock ``Pseudorandom unitaries with non-adaptive security''.
\newblock \href{https://dx.doi.org/10.48550/arXiv.2402.14803}{{a}rXiv preprint
  {a}rXiv:2402.14803}~(2024).

\bibitem{brakerski2024real}
Zvika Brakerski and Nir Magrafta.
\newblock ``Real-valued somewhat-pseudorandom unitaries''.
\newblock \href{https://dx.doi.org/10.48550/arXiv.2403.16704}{{a}rXiv preprint
  {a}rXiv:2403.16704}~(2024).

\bibitem{metger2024simple}
Tony Metger, Alexander Poremba, Makrand Sinha, and Henry Yuen.
\newblock ``Simple constructions of linear-depth t-designs and pseudorandom
  unitaries''.
\newblock \href{https://dx.doi.org/10.48550/arXiv.2404.12647}{{a}rXiv preprint
  {a}rXiv:2404.12647}~(2024).

\bibitem{chitambar2019quantum}
Eric Chitambar and Gilad Gour.
\newblock ``Quantum resource theories''.
\newblock \href{https://dx.doi.org/10.1103/RevModPhys.91.025001}{Rev. Mod.
  Phys. {\bf 91}, 025001}~(2019).

\bibitem{veitch2014resource}
Victor Veitch, S~A~Hamed Mousavian, Daniel Gottesman, and Joseph Emerson.
\newblock ``The resource theory of stabilizer quantum computation''.
\newblock \href{https://dx.doi.org/10.1088/1367-2630/16/1/013009}{New Journal
  of Physics {\bf 16}, 013009}~(2014).

\bibitem{barenco1997stabilization}
Adriano Barenco, Andre Berthiaume, David Deutsch, Artur Ekert, Richard Jozsa,
  and Chiara Macchiavello.
\newblock ``Stabilization of quantum computations by symmetrization''.
\newblock \href{https://dx.doi.org/10.1137/S0097539796302452}{SIAM Journal on
  Computing {\bf 26}, 1541--1557}~(1997).

\bibitem{ekert2002direct}
Artur~K. Ekert, Carolina~Moura Alves, Daniel K.~L. Oi, Micha\l{} Horodecki,
  Pawe\l{} Horodecki, and L.~C. Kwek.
\newblock ``Direct estimations of linear and nonlinear functionals of a quantum
  state''.
\newblock \href{https://dx.doi.org/10.1103/PhysRevLett.88.217901}{Phys. Rev.
  Lett. {\bf 88}, 217901}~(2002).

\bibitem{huang2021demonstrating}
Hsin-Yuan Huang, Michael Broughton, Jordan Cotler, Sitan Chen, Jerry Li, Masoud
  Mohseni, Hartmut Neven, Ryan Babbush, Richard Kueng, John Preskill, et~al.
\newblock ``Quantum advantage in learning from experiments''.
\newblock \href{https://dx.doi.org/10.1126/science.abn7293}{Science {\bf 376},
  1182--1186}~(2022).

\bibitem{nielsen2002quantum}
Michael~A Nielsen and Isaac~L Chuang.
\newblock ``Quantum computation and quantum information''.
\newblock \href{https://dx.doi.org/10.1017/CBO9780511976667}{Cambridge
  University Press}. ~(2010).

\bibitem{rudolph20022}
Terry Rudolph and Lov Grover.
\newblock ``A 2 rebit gate universal for quantum computing''.
\newblock \href{https://dx.doi.org/10.48550/arXiv.quant-ph/0210187}{{a}rXiv
  preprint quant-ph/0210187}~(2002).

\bibitem{hickey2018quantifying}
Alexander Hickey and Gilad Gour.
\newblock ``Quantifying the imaginarity of quantum mechanics''.
\newblock \href{https://dx.doi.org/10.1088/1751-8121/aabe9c}{Journal of Physics
  A: Mathematical and Theoretical {\bf 51}, 414009}~(2018).

\bibitem{giurgica2023pseudorandomness}
Tudor Giurgica-Tiron and Adam Bouland.
\newblock ``Pseudorandomness from subset states''.
\newblock \href{https://dx.doi.org/10.48550/arXiv.2312.09206}{{a}rXiv preprint
  {a}rXiv:2312.09206}~(2023).

\bibitem{jeronimo2024pseudorandom}
Fernando~Granha Jeronimo, Nir Magrafta, and Pei Wu.
\newblock ``Pseudorandom and pseudoentangled states from subset states''.
\newblock \href{https://dx.doi.org/10.48550/arXiv.2312.15285}{{a}rXiv preprint
  {a}rXiv:2312.15285}~(2023).

\bibitem{gottesman1998heisenberg}
Daniel Gottesman.
\newblock ``The {H}eisenberg representation of quantum computers''.
\newblock Group22: Proceedings of the XXII International Colloquium on Group
  Theoretical Methods in PhysicsPages 32--43~(1999).
\newblock
  url:~\href{https://doi.org/10.48550/arXiv.quant-ph/9807006}{doi.org/10.48550/arXiv.quant-ph/9807006}.

\bibitem{aharonov2021quantum}
Dorit Aharonov, Jordan Cotler, and Xiao-Liang Qi.
\newblock ``Quantum algorithmic measurement''.
\newblock \href{https://dx.doi.org/10.1038/s41467-021-27922-0}{Nature
  communications {\bf 13}, 887}~(2022).

\bibitem{singh2016average}
Uttam Singh, Lin Zhang, and Arun~Kumar Pati.
\newblock ``Average coherence and its typicality for random pure states''.
\newblock \href{https://dx.doi.org/10.1103/PhysRevA.93.032125}{Phys. Rev. A
  {\bf 93}, 032125}~(2016).

\bibitem{gu2023little}
Andi Gu, Lorenzo Leone, Soumik Ghosh, Jens Eisert, Susanne~F. Yelin, and Yihui
  Quek.
\newblock ``Pseudomagic quantum states''.
\newblock \href{https://dx.doi.org/10.1103/PhysRevLett.132.210602}{Phys. Rev.
  Lett. {\bf 132}, 210602}~(2024).

\bibitem{koh2018quantum}
Dax~Enshan Koh, Murphy~Yuezhen Niu, and Theodore~J Yoder.
\newblock ``Quantum simulation from the bottom up: the case of rebits''.
\newblock \href{https://dx.doi.org/10.1088/1751-8121/aab9c4}{Journal of Physics
  A: Mathematical and Theoretical {\bf 51}, 195302}~(2018).

\bibitem{mckague2013power}
Matthew McKague.
\newblock ``On the power quantum computation over real {H}ilbert spaces''.
\newblock \href{https://dx.doi.org/10.1142/S0219749913500019}{International
  Journal of Quantum Information {\bf 11}, 1350001}~(2013).

\bibitem{delfosse2015wigner}
Nicolas Delfosse, Philippe Allard~Guerin, Jacob Bian, and Robert Raussendorf.
\newblock ``Wigner function negativity and contextuality in quantum computation
  on rebits''.
\newblock \href{https://dx.doi.org/10.1103/PhysRevX.5.021003}{Phys. Rev. X {\bf
  5}, 021003}~(2015).

\bibitem{yang2014certifying}
Yuxiang Yang, Giulio Chiribella, and Gerardo Adesso.
\newblock ``Certifying quantumness: Benchmarks for the optimal processing of
  generalized coherent and squeezed states''.
\newblock \href{https://dx.doi.org/10.1103/PhysRevA.90.042319}{Phys. Rev. A
  {\bf 90}, 042319}~(2014).

\bibitem{miyazaki2019complex}
Jisho Miyazaki, Akihito Soeda, and Mio Murao.
\newblock ``Complex conjugation supermap of unitary quantum maps and its
  universal implementation protocol''.
\newblock \href{https://dx.doi.org/10.1103/PhysRevResearch.1.013007}{Phys. Rev.
  Res. {\bf 1}, 013007}~(2019).

\bibitem{katz2020introduction}
Jonathan Katz and Yehuda Lindell.
\newblock ``Introduction to modern cryptography''.
\newblock \href{https://dx.doi.org/10.1201/b17668}{CRC press}. ~(2020).

\bibitem{bharti2021noisy}
Kishor Bharti, Alba Cervera-Lierta, Thi~Ha Kyaw, Tobias Haug, Sumner
  Alperin-Lea, Abhinav Anand, Matthias Degroote, Hermanni Heimonen, Jakob~S.
  Kottmann, Tim Menke, Wai-Keong Mok, Sukin Sim, Leong-Chuan Kwek, and Al\'an
  Aspuru-Guzik.
\newblock ``Noisy intermediate-scale quantum algorithms''.
\newblock \href{https://dx.doi.org/10.1103/RevModPhys.94.015004}{Rev. Mod.
  Phys. {\bf 94}, 015004}~(2022).

\bibitem{babbush2021focus}
Ryan Babbush, Jarrod~R McClean, Michael Newman, Craig Gidney, Sergio Boixo, and
  Hartmut Neven.
\newblock ``Focus beyond quadratic speedups for error-corrected quantum
  advantage''.
\newblock \href{https://dx.doi.org/10.1103/PRXQuantum.2.010103}{PRX Quantum
  {\bf 2}, 010103}~(2021).

\bibitem{suzuki2022quantum}
Yasunari Suzuki, Suguru Endo, Keisuke Fujii, and Yuuki Tokunaga.
\newblock ``Quantum error mitigation as a universal error reduction technique:
  applications from the {NISQ} to the fault-tolerant quantum computing eras''.
\newblock \href{https://dx.doi.org/10.1103/PRXQuantum.3.010345}{PRX Quantum
  {\bf 3}, 010345}~(2022).

\bibitem{grewal2024pseudoentanglement}
Sabee Grewal, Vishnu Iyer, William Kretschmer, and Daniel Liang.
\newblock ``Pseudoentanglement ain't cheap''.
\newblock \href{https://dx.doi.org/10.48550/arXiv.2404.00126}{{a}rXiv preprint
  {a}rXiv:2404.00126}~(2024).

\bibitem{wu2023quantum}
Ya-Dong Wu, Yan Zhu, Giulio Chiribella, and Nana Liu.
\newblock ``Efficient learning of continuous-variable quantum states''.
\newblock \href{https://dx.doi.org/10.1103/PhysRevResearch.6.033280}{Phys. Rev.
  Res. {\bf 6}, 033280}~(2024).

\bibitem{king2024exponential}
Robbie King, Kianna Wan, and Jarrod~R. McClean.
\newblock ``Exponential learning advantages with conjugate states and minimal
  quantum memory''.
\newblock \href{https://dx.doi.org/10.1103/PRXQuantum.5.040301}{PRX Quantum
  {\bf 5}, 040301}~(2024).

\bibitem{bendersky2009general}
Ariel Bendersky, Juan~Pablo Paz, and Marcelo~Terra Cunha.
\newblock ``General theory of measurement with two copies of a quantum state''.
\newblock \href{https://dx.doi.org/10.1103/PhysRevLett.103.040404}{Physical
  review letters {\bf 103}, 040404}~(2009).

\bibitem{haug2022scalable}
Tobias Haug and M.S. Kim.
\newblock ``Scalable measures of magic resource for quantum computers''.
\newblock \href{https://dx.doi.org/10.1103/PRXQuantum.4.010301}{PRX Quantum
  {\bf 4}, 010301}~(2023).

\bibitem{zhandry2012how}
Mark Zhandry.
\newblock ``How to construct quantum random functions''.
\newblock \href{https://dx.doi.org/10.1109/FOCS.2012.37}{Journal of the ACM
  (JACM) {\bf 68}, 1--43}~(2021).

\bibitem{popescu2005foundations}
Sandu Popescu, Anthony~J Short, and Andreas Winter.
\newblock ``Entanglement and the foundations of statistical mechanics''.
\newblock \href{https://dx.doi.org/10.1038/nphys444}{Nature Physics {\bf 2},
  754--758}~(2006).

\bibitem{bae2015quantum}
Joonwoo Bae and Leong-Chuan Kwek.
\newblock ``Quantum state discrimination and its applications''.
\newblock \href{https://dx.doi.org/10.1088/1751-8113/48/8/083001}{Journal of
  Physics A: Mathematical and Theoretical {\bf 48}, 083001}~(2015).

\bibitem{montanaro2017learning}
Ashley Montanaro.
\newblock ``Learning stabilizer states by {B}ell sampling''.
\newblock \href{https://dx.doi.org/10.48550/arXiv.1707.04012}{{a}rXiv preprint
  {a}rXiv:1707.04012}~(2017).

\bibitem{guhne2009entanglement}
Otfried Gühne and Géza Tóth.
\newblock ``Entanglement detection''.
\newblock
  \href{https://dx.doi.org/https://doi.org/10.1016/j.physrep.2009.02.004}{Physics
  Reports {\bf 474}, 1--75}~(2009).

\bibitem{khatri2019quantum}
Sumeet Khatri, Ryan LaRose, Alexander Poremba, Lukasz Cincio, Andrew~T.
  Sornborger, and Patrick~J. Coles.
\newblock ``Quantum-assisted quantum compiling''.
\newblock \href{https://dx.doi.org/10.22331/q-2019-05-13-140}{{Quantum} {\bf
  3}, 140}~(2019).

\bibitem{cacciapuoti2020when}
Angela~Sara Cacciapuoti, Marcello Caleffi, Rodney Van~Meter, and Lajos Hanzo.
\newblock ``When entanglement meets classical communications: Quantum
  teleportation for the quantum internet''.
\newblock \href{https://dx.doi.org/10.1109/TCOMM.2020.2978071}{IEEE
  Transactions on Communications {\bf 68}, 3808--3833}~(2020).

\bibitem{xing2023fundamental}
Junjing Xing, Tianfeng Feng, Zhaobing Fan, Haitao Ma, Kishor Bharti, Dax~Enshan
  Koh, and Yunlong Xiao.
\newblock ``Fundamental limitations on communication over a quantum network''.
\newblock \href{https://dx.doi.org/10.48550/arXiv.2306.04983}{{a}rXiv preprint
  {a}rXiv:2306.04983}~(2023).

\bibitem{ekert1991quantum}
Artur~K. Ekert.
\newblock ``Quantum cryptography based on {B}ell's theorem''.
\newblock \href{https://dx.doi.org/10.1103/PhysRevLett.67.661}{Phys. Rev. Lett.
  {\bf 67}, 661--663}~(1991).

\bibitem{briegel2009measurement}
Hans~J Briegel, David~E Browne, Wolfgang D{\"u}r, Robert Raussendorf, and
  Maarten Van~den Nest.
\newblock ``Measurement-based quantum computation''.
\newblock \href{https://dx.doi.org/10.1038/nphys1157}{Nature Physics {\bf 5},
  19--26}~(2009).

\bibitem{biamonte2020entanglement}
Jacob~D. Biamonte, Mauro E.~S. Morales, and Dax~Enshan Koh.
\newblock ``Entanglement scaling in quantum advantage benchmarks''.
\newblock \href{https://dx.doi.org/10.1103/PhysRevA.101.012349}{Phys. Rev. A
  {\bf 101}, 012349}~(2020).

\bibitem{jozsa2014classical}
Richard Jozsa and Maarten Van~den Nest.
\newblock ``Classical simulation complexity of extended {C}lifford circuits''.
\newblock \href{https://dx.doi.org/10.26421/QIC14.7-8}{Quantum Information \&
  Computation {\bf 14}, 633--648}~(2014).

\bibitem{koh2015further}
Dax~Enshan Koh.
\newblock ``Further extensions of {C}lifford circuits and their classical
  simulation complexities''.
\newblock \href{https://dx.doi.org/10.26421/QIC17.3-4}{Quantum Information \&
  Computation {\bf 17}, 0262--0282}~(2017).

\bibitem{bouland2018complexity}
Adam Bouland, Joseph~F. Fitzsimons, and Dax~Enshan Koh.
\newblock ``{Complexity Classification of Conjugated Clifford Circuits}''.
\newblock In Rocco~A. Servedio, editor, 33rd Computational Complexity
  Conference (CCC 2018).
\newblock \href{https://dx.doi.org/10.4230/LIPIcs.CCC.2018.21}{Volume 102 of
  Leibniz International Proceedings in Informatics (LIPIcs), pages
  21:1--21:25}.
\newblock Dagstuhl, Germany~(2018). Schloss Dagstuhl--Leibniz-Zentrum f{\"u}r
  Informatik.

\bibitem{haug2023efficient}
Tobias Haug, Soovin Lee, and M.~S. Kim.
\newblock ``Efficient quantum algorithms for stabilizer entropies''.
\newblock \href{https://dx.doi.org/10.1103/PhysRevLett.132.240602}{Phys. Rev.
  Lett. {\bf 132}, 240602}~(2024).

\bibitem{bravyi2019simulation}
Sergey Bravyi, Dan Browne, Padraic Calpin, Earl Campbell, David Gosset, and
  Mark Howard.
\newblock ``Simulation of quantum circuits by low-rank stabilizer
  decompositions''.
\newblock \href{https://dx.doi.org/10.22331/q-2019-09-02-181}{{Quantum} {\bf
  3}, 181}~(2019).

\bibitem{bu2019efficient}
Kaifeng Bu and Dax~Enshan Koh.
\newblock ``Efficient classical simulation of {C}lifford circuits with
  nonstabilizer input states''.
\newblock \href{https://dx.doi.org/10.1103/PhysRevLett.123.170502}{Phys. Rev.
  Lett. {\bf 123}, 170502}~(2019).

\bibitem{seddon2021quantifying}
James~R. Seddon, Bartosz Regula, Hakop Pashayan, Yingkai Ouyang, and Earl~T.
  Campbell.
\newblock ``Quantifying quantum speedups: Improved classical simulation from
  tighter magic monotones''.
\newblock \href{https://dx.doi.org/10.1103/PRXQuantum.2.010345}{PRX Quantum
  {\bf 2}, 010345}~(2021).

\bibitem{leone2021renyi}
Lorenzo Leone, Salvatore F.~E. Oliviero, and Alioscia Hamma.
\newblock ``Stabilizer {R}\'enyi entropy''.
\newblock \href{https://dx.doi.org/10.1103/PhysRevLett.128.050402}{Phys. Rev.
  Lett. {\bf 128}, 050402}~(2022).

\bibitem{bu2022statistical}
Kaifeng Bu, Dax~Enshan Koh, Lu~Li, Qingxian Luo, and Yaobo Zhang.
\newblock ``Statistical complexity of quantum circuits''.
\newblock \href{https://dx.doi.org/10.1103/PhysRevA.105.062431}{Phys. Rev. A
  {\bf 105}, 062431}~(2022).

\bibitem{bu2023stabilizer}
Kaifeng Bu, Weichen Gu, and Arthur Jaffe.
\newblock ``Stabilizer testing and magic entropy''.
\newblock \href{https://dx.doi.org/10.48550/arXiv.2306.09292}{{a}rXiv preprint
  {a}rXiv:2306.09292}~(2023).

\bibitem{bu2023magic}
Kaifeng Bu, Weichen Gu, and Arthur Jaffe.
\newblock ``Discrete quantum {G}aussians and central limit theorem''.
\newblock \href{https://dx.doi.org/10.48550/arXiv.2302.08423}{{a}rXiv preprint
  {a}rXiv:2302.08423}~(2023).

\bibitem{bu2023quantum}
Kaifeng Bu, Weichen Gu, and Arthur Jaffe.
\newblock ``Quantum entropy and central limit theorem''.
\newblock \href{https://dx.doi.org/10.1073/pnas.2304589120}{Proceedings of the
  National Academy of Sciences {\bf 120}, e2304589120}~(2023).

\bibitem{low2009learning}
Richard~A. Low.
\newblock ``Learning and testing algorithms for the {C}lifford group''.
\newblock \href{https://dx.doi.org/10.1103/PhysRevA.80.052314}{Phys. Rev. A
  {\bf 80}, 052314}~(2009).

\bibitem{wang2011property}
Guoming Wang.
\newblock ``Property testing of unitary operators''.
\newblock \href{https://dx.doi.org/10.1103/PhysRevA.84.052328}{Phys. Rev. A
  {\bf 84}, 052328}~(2011).

\bibitem{collins2006integration}
Beno{\^\i}t Collins and Piotr {\'S}niady.
\newblock ``Integration with respect to the haar measure on unitary, orthogonal
  and symplectic group''.
\newblock \href{https://dx.doi.org/10.1007/s00220-006-1554-3}{Communications in
  Mathematical Physics {\bf 264}, 773--795}~(2006).

\bibitem{puchala2017symbolic}
Z.~Puchała and J.A. Miszczak.
\newblock ``Symbolic integration with respect to the {H}aar measure on the
  unitary groups''.
\newblock \href{https://dx.doi.org/10.1515/bpasts-2017-0003}{Bulletin of the
  Polish Academy of Sciences: Technical Sciences {\bf 65}, 21--27}~(2017).

\end{thebibliography}

\let\addcontentsline\oldaddcontentsline

\onecolumngrid
\newpage 

\setcounter{secnumdepth}{2}
\setcounter{equation}{0}
\setcounter{figure}{0}
\setcounter{section}{0}
\renewcommand{\thetable}{S\arabic{table}}
\renewcommand{\theequation}{S\arabic{equation}}
\renewcommand{\thefigure}{S\arabic{figure}}

\renewcommand{\thesection}{\Alph{section}}
\renewcommand{\thesubsection}{\arabic{subsection}}
\renewcommand*{\theHsection}{\thesection}

\clearpage
\begin{center}

\textbf{\large Appendix}
\end{center}
\setcounter{equation}{0}
\setcounter{figure}{0}
\setcounter{table}{0}

\makeatletter
\renewcommand{\theequation}{S\arabic{equation}}
\renewcommand{\thefigure}{S\arabic{figure}}
\newtheorem{thmS}{Theorem S\ignorespaces}

\newtheorem{claimS}{Claim S\ignorespaces}

\newtheorem{definitionS}{Definition S\ignorespaces}

\newtheorem{lemS}{Lemma S\ignorespaces}
\newtheorem{propositionS}{Proposition S\ignorespaces}

We provide proofs and additional details supporting the claims in the main text.

\makeatletter
\@starttoc{toc}
\makeatother

\section{Primer on PRSs, PRSSs and PRUs} \label{sec:basics}
We start with the basics of classical cryptography required to understand PRSs, PRSSs and PRUs. For details, refer to Ref.~\cite{katz2020introduction}.
\subsection{Classical cryptography basics}
First, we start with the concepts of negligible and noticeable functions.

\begin{definitionS}[Negligible Function]
A  function $\mu:\mathbb{N} \to \mathbb R$ is negligible if and only if $\forall c \in \mathbb{N}$, $\exists n_0 \in \mathbb{N}$ such that $\forall n > n_0$, $\mu(n) < n^{-c}$.    
\end{definitionS}

\begin{definitionS}[Noticeable Function]
A  function 
$\mu:\mathbb{N} \to \mathbb R$ is noticeable if and only if $\exists c \in \mathbb{N}$, $n_0 \in \mathbb{N}$ such that $\forall n \ge n_0$, $\mu(n) \ge n^{-c}$.
\end{definitionS}
\revC{Intuitively, a noticeable function is a function that at least scales as a polynomial $n^{-c}$ with $c>0$ being a constant.
In contrast, a negligible function must be strictly smaller than any polynomial function. }
For example, since $\mu(n) = 2^{-n}$ is exponentially small, it is a negligible function. Conversely, $\mu(n) = n^{-3}$ is only polynomially small and is therefore noticeable. 

Finally, a pseudorandom generator (PRG) is a deterministic algorithm that takes a short, truly random input called a seed and expands it into a longer sequence of seemingly random bits. The generated sequence should be indistinguishable from a truly random sequence to any computationally bounded observer or algorithm.

\begin{definitionS}[Pseudorandom Generator (PRG)]
Let $l(\cdot)$ be a polynomial function and let $G$ be a deterministic polynomial-time algorithm. $G$ takes an input $s$ from the set $\{0,1\}^n$ and produces a string of length $l(n)$. We define $G$ as a pseudorandom generator if it satisfies the following two conditions:
\begin{enumerate}
    \item \emph{Expansion:} For every $n$, the length of the output string, $l(n)$, is greater than $n$.
    \item \emph{Pseudorandomness:} For any probabilistic polynomial-time distinguisher $D$, there exists a negligible function $\mathrm{negl}(n)$ such that
    $$\vert \mathrm{Pr}[D(r)=1] - \mathrm{Pr}[D(G(s))=1] \vert \leq \mathrm{negl}(n). $$
    Here, $r$ is uniformly chosen at random from $\{0,1\}^{l(n)}$, the seed $s$ is uniformly chosen at random from $\{0, 1\}^n$, and the probabilities are computed over the random coins used by $D$ and the choices of $r$ and $s$.
\end{enumerate}
The function $l(\cdot)$ is referred to as the expansion factor of the pseudorandom generator $G$.
\end{definitionS}

Pseudorandom generators play a crucial role in cryptography, as they provide a way to generate seemingly random values from a small random seed. They are used in a variety of applications, such as key generation, encryption schemes, and secure communication protocols. PRGs are also used in other areas of computer science, including simulations, random sampling, and probabilistic algorithms. Finally, to illustrate the construction of a PRS 
from a PRG, we discuss  the concept of pseudorandom functions (PRF):
\begin{definitionS}[Pseudorandom Function (PRF)]
\revC{A keyed family of functions %
$\{F_k: \{0, 1\}^n \rightarrow \{0, 1\}^n\}_{k\in\mathcal{K}}$}
is called a pseudorandom function if $F_k$ is computable in polynomial time and for all probabilistic polynomial-time distinguishers $D$, there exists a negligible function $\mathrm{negl}$ such that:
$$ \vert \Pr[D^{F_k(\cdot)}(1^n)=1] - \Pr[D^{f(\cdot)}(1^n)=1]  \vert \leq \mathrm{negl}(n) $$
where $k$ is chosen uniformly at random from $\mathcal{K}=\{0, 1\}^n$, and $f$ is chosen uniformly at random from the set of all functions which map \revC{$n$-bit strings to $n$-bit strings}.
\end{definitionS}
Without knowing the key $k$, it should be computationally infeasible for an adversary to distinguish between the output of the PRF and a truly random function, even if the adversary can make adaptive queries.

\revC{PRFs and PRGs are equivalent primitives in cryptography, i.e. PRF $\iff$ PRG~\cite{katz2020introduction}. This extends to the case where the distinguisher has access to quantum computer, i.e. post-quantum PRF $\iff$ post-quantum PRG~\cite{zhandry2012how}. Furthermore, Ref.~\cite{ji2018pseudorandom} proved that post-quantum PRF $\implies$ PRS.}

\subsection{Definitions of pseudorandomness}
\begin{definitionS}[Pseudorandom Unitary Operators (PRUs)~\cite{ji2018pseudorandom}]
Consider a $n$-qubit Hilbert space $\mathcal{H}$ and a key space $\mathcal{K}$ with security parameter $\kappa=\operatorname{poly}(n)$.
A family of unitary operators $\{U_k \in U(\mathcal{H})\}_{k \in \mathcal{K}}$ is pseudorandom if the following two conditions hold:

\begin{enumerate}
\item \textit{Efficient computation}: There exists an efficient quantum algorithm $Q$ such that for all $k$ and any $\ket{\psi} \in \mathcal{S}(\mathcal{H})$, $Q(k, \ket{\psi}) = U_k \ket{\psi}$.
\item \textit{Pseudorandomness}: $U_k$ with a random key $k$ is computationally indistinguishable from a Haar random unitary operator. More precisely, for any efficient quantum algorithm $A$ that makes at most polynomially many queries to the oracle,
\[
\left| \Pr_{k \gets \mathcal{K}}[A^{U_k}(1^\kappa) = 1] - \Pr_{U \gets \mu}[A^U(1^\kappa) = 1] \right| \leq \operatorname{negl}(n).
\]
\end{enumerate}
\end{definitionS}
\revA{In the first term, the observer uses any efficient algorithm $A$ with access to the key size $1^{\kappa}$ in unary and oracle access to pseudorandom unitary $U_k$ with random key $k$, while in the second term the observer is given Haar random unitaries $U$.}
This definition states that a family of unitary operators is considered pseudorandom if it is both efficiently computable and indistinguishable from Haar random unitary operators for an observer with a bounded computational power.

\begin{definitionS}[Pseudorandom Quantum States (PRSs)]\label{def:PRS}
Consider an $n$-qubit Hilbert space $\mathcal{H}$ and a key space $\mathcal{K}$ with security parameter $\kappa=\operatorname{poly}(n)$. A keyed family of quantum states ${ \ket{\phi_k} \in \mathcal{S}(\mathcal{H}) }_{k \in \mathcal{K}}$ is defined as pseudorandom if it satisfies the following conditions:

\begin{enumerate}
\item \textit{Efficient generation}: There exists a polynomial-time quantum algorithm $G$ capable of generating the state $\ket{\phi_k}$ when given the input $k$. In other words, for every $k \in \mathcal{K}$, $G(k) = \ket{\phi_k}$.
\item \textit{Pseudorandomness}: When given the same random $k \in \mathcal{K}$, any polynomially bounded number of copies of $\ket{\phi_k}$ are computationally indistinguishable from the same number of copies of a Haar random state. More specifically, for any efficient quantum algorithm $A$ and any $m \in \operatorname{poly}(n)$,
\[\left| \Pr_{k \gets \mathcal{K}} [A(\ket{\phi_k}^{\otimes m}) = 1] - \Pr_{\ket{\psi} \gets \mu} [A(\ket{\psi}^{\otimes m}) = 1] \right| = \operatorname{negl}(n),\]
where $\mu$ represents the Haar measure on $\mathcal{S}(\mathcal{H})$.
\end{enumerate}
\end{definitionS}
In this definition, a keyed family of quantum states is considered pseudorandom if it can be generated efficiently and appears statistically indistinguishable from Haar random states to an observer with limited computational resources.

\revA{In between PRUs and PRS are pseudorandom state scramblers (PRSSs). They are an ensemble of $n$-qubit unitaries $\{U_i\}_i$ which applied to any (fixed) $n$-qubit input state $\ket{\psi}$ generate a PRS $\{U_i \ket{\psi}\}_i$~\cite{lu2023quantum,ananth2024pseudorandom}. 
Formally, we have~\cite{lu2023quantum}:
\begin{definitionS}[Pseudorandom State Scrambler (PRSS)]
  \label{def:prssb}
  Let $\kappa=\operatorname{poly}(n)$ be a security
  parameter. A \emph{pseudorandom state scrambler} (PRSS) is an ensemble of $n$-qubit unitary operators $\{ U_k \in \mathcal{U}(\mathcal{H}) \}_{k\in\mathcal{K}}$ 
satisfying:
  \begin{enumerate}
  \item \textit{Efficient computation}: There exists an efficient quantum algorithm $G$, such that for all $k$ and any $\ket{\psi} \in \mathcal{S}(\mathcal{H})$, $G(k, \ket{\psi}) = U_k \ket{\psi}$.
  \item Pseudorandomness: For any $m = \mathrm{poly}(n)$, any $n$-qubit state $\ket{\phi}\in \mathcal{S}(\mathcal{H})$, and any efficient quantum algorithm $A$ we have
\[\left| \Pr_{k \gets \mathcal{K}} [A((U_k\ket{ \phi})^{\otimes m}) = 1] - \Pr_{\ket{\psi} \gets \mu} [A(\ket{\psi}^{\otimes m}) = 1] \right| = \operatorname{negl}(n),\]
where $\mu$ is the Haar measure on $\mathcal{S}(\mathcal{H})$.

  \end{enumerate}

\end{definitionS}
Note that while PRSS are unitaries that generate PRS, not all unitary generators of PRS are also PRSS. In particular, there are PRS generators that produce PRS only when applied to one particular input state~\cite{ji2018pseudorandom}. For example, the recently proposed construction of the binary phase state requires a product state as input, while for general input states the same generators will not produce PRS~\cite{brakerski2019pseudo}. 
Also note that any PRU is a PRSS, while the converse is not true in general as we show in the main text.
While not discussed here, note that one can define more general notion of PRSS for isometries~\cite{lu2023quantum,ananth2024pseudorandom}.}

Finally, we discuss the notion of pseudoentanglement, as introduced in Ref.~\cite{bouland2022quantum}.
\begin{definitionS}[Pseudoentangled State Ensemble (PES)]
A \emph{pseudoentangled state ensemble (PES)} with gap $f(n)$ vs. $g(n)$ consists of two ensembles of $n$-qubit states $\ket{\Psi_k}$ and $\ket{\Phi_k}$, indexed by a secret key $k \in \mathcal{K}$ with the following properties:

\begin{enumerate}
\item \emph{Efficient Preparation}: Given $k$, $\ket{\Psi_k}$ (or $\ket{\Phi_k}$, respectively) is efficiently preparable by a uniform, poly-sized quantum circuit.

\item \emph{Entanglement Entropy}: With probability $\geq 1 - \frac{1}{\text{poly}(n)}$ over the choice of $k$, the entanglement entropy between the first $\frac{n}{2}$ and second $\frac{n}{2}$ qubits of $\ket{\Psi_k}$ (or $\ket{\Phi_k}$, respectively) is $\Theta(f(n))$ (or $\Theta(g(n))$, respectively).

\item \emph{Indistinguishability}: \revC{No polynomial-time quantum algorithm with access to $m=\mathrm{poly}(n)$ copies can distinguish between the two ensembles}  %
with more than negligible probability. That is, for any poly-time quantum algorithm $A$, we have that
\[\left| \Pr_{k \gets \mathcal{K}} [A(\ket{\Phi_k}^{\otimes m}) = 1] - \Pr_{k \gets \mathcal{K}} [A(\ket{\Psi_k}^{\otimes m}) = 1] \right| = \operatorname{negl}(n)\,.\]
\end{enumerate}
\end{definitionS}

We now generalize this concept to arbitrary resource theories:
\begin{definitionS}[Pseudoresource state ensemble]\label{def:pseudoresource}
With respect to resource $Q$,  pseudoresource state ensemble with
gap $f(n)$ vs. $g(n)$ (where $f(n)\ge g(n)$) consists of two ensembles with key $k\in\mathcal{K}$: A `high-resource' ensemble $\ket{\psi_k}$ which has $Q(\ket{\psi_k})=f(n)$ with high probablity over $k$, and `low-resource' ensemble $\ket{\phi_k}$ with $Q(\ket{\phi_k})=g(n)$. We demand following conditions for the pseudoresource state ensemble:
\begin{itemize}
\item $\ket{\psi_k}$ and $\ket{\phi_k}$ are efficiently preparable
\item No poly-time quantum algorithm $A$ can distinguish between the ensembles $\rho=\ket{\psi_k}^{\otimes m}\bra{\psi_k}^{\otimes m}$ and $\sigma=\ket{\phi_k}^{\otimes m}\bra{\phi_k}^{\otimes m}$ with $m\in\mathrm{poly}(n)$
with more than negligible probability:
\[\left| \Pr_{k \gets \mathcal{K}} [A(\ket{\psi_k}^{\otimes m}) = 1] - \Pr_{k \gets \mathcal{K}} [A(\ket{\phi_k}^{\otimes m}) = 1] \right| = \operatorname{negl}(n)\,.\]
 \end{itemize}
\end{definitionS}

\subsection{The indistinguishability experiment}
It is crucial to emphasize that a pseudorandom object possesses any efficiently verifiable property that its random counterpart exhibits. To illustrate this concept, we consider a two-player experiment $\mathcal{E}\left(\adversary,n  \right)$ involving a challenger $\challenger$ and a quantum polynomial time adversary $\adversary$. Let ${ \ket{\phi_k} \in \mathcal{S}(\mathcal{H}) }_{k \in \mathcal{K}}$  represent a  keyed family of $n$-qubit quantum states and $\mu$ represent the Haar measure on $\mathcal{S}(\mathcal{H})$. Let $m = \mathrm{poly}(n)$. \revA{The key generator algorithm $\mathrm{Gen}(1^n)$ takes the security parameter $n$ in unary as input, and outputs key $k$, which is a uniform random bitstring of size $\mathrm{poly}(n)$.} The experiment $\mathcal{E}\left(\adversary,n  \right)$  proceeds as follows.
\begin{enumerate}
    \item $\challenger$ generates a private key $k \leftarrow \Gen(1^n)$, where $\vert k \vert = \mathrm{poly}(n)$.
    \item $\challenger$ samples uniform random a bit $b \leftarrow \bit$.
    \begin{enumerate}
        \item If b = 0, $\challenger$ prepares a state $\ket{\psi_0} := \ket{\phi_k}^{\otimes m}$.
        \item If b = 1, $\challenger$ prepares a state $\ket{\psi_1} := \ket{\phi_\mu}^{\otimes m}$. Here, $\ket{\phi_\mu}$ is a Haar random state.
    \end{enumerate}
    \item $\challenger$ sends $\ket{\psi_b}$ to $\adversary$.
    \item $\adversary$ sends $b^{\prime} \in \bit$ to $\challenger$.
    \item The output of the experiment is defined to be $1$ if $b^{\prime} = b$, else $0$. If $\mathcal{E}\left(\adversary,n  \right)=1$, we say $\adversary$ succeeded.
\end{enumerate}
If ${ \ket{\phi_k} \in \mathcal{S}(\mathcal{H}) }_{k \in \mathcal{K}}$  is PRS, then $\text{Pr}[ \mathcal{E}\left(\adversary,n  \right)=1 ]  \leq \frac{1}{2} + \text{negl}(n)$. Here, the probability is taken over all the random coins used in the experiment. If $\ket{\phi_k}$ does not possess any efficiently verifiable property that its random counterpart exhibits, then it can be used by $\adversary$ to win the indistinguishability game $ \mathcal{E}\left(\adversary,n  \right)$ with a probability greater than $\frac{1}{2} + \text{negl}(n)$.

\section{Property Testing}\label{sec:propert_test}

\emphsection{Property testing}%
A property tester $\mathcal{A}_Q$ has to fulfill two conditions~\cite{rubinfeld1996robust,goldreich1998property,buhrman2008quantum,montanaro2013survey}: The completeness condition demands that the tester accepts with high probability if the state has the property within threshold $\beta$. The soundness condition states that the tester rejects with high probability if the state exceeds a threshold value $\delta$ of the property.

\begin{definitionS}[Property tester]\label{def:prop-tester}
An algorithm $\mathcal{A}_Q$ is a  tester for property $Q$ using $t=t(n,\delta,\beta)$ copies if, given $t$ copies of $\ket{\psi}\in\mathbb{C}^{2^n}$, constants $\beta>0$ and $\delta>\beta$,  it acts as follows:
\begin{itemize}
\item (Completeness) If $Q(\ket{\psi})\le \beta$, then
\begin{align}
\Pr[\mathcal{A}_Q\text{ accepts given }\ket{\psi}^{\otimes t}]\geq\frac{2}{3}.
\end{align}
\item (Soundness) If $Q(\ket{\psi})\ge \delta$, then 
\begin{align}
\Pr[\mathcal{A}_Q\text{ accepts given $\ket{\psi}^{\otimes t}$}]\leq\frac{1}{3}.
\end{align}
\end{itemize}
\end{definitionS}

\section{PRSs, PRUs and noise}\label{sec:noise}
Here, we show that PRSs and PRUs (and also PRSSs) can sustain only a negligible amount of noise.
First, we restate Theorem~\ref{thm:noise} of the main text:
\begin{thmS}[PRUs and PRSs are not robust to noise ]\label{thm:noise_sup}
Any ensemble of unitaries or states subject to a single-qubit depolarizing noise channel $\Lambda_p(X)=(1-p) X+p\tr(X)I_1/2$ can be PRUs and PRSs only for at most $p=\mathrm{negl}(n)$ depolarizing probability.
\end{thmS}
\begin{proof}
The purity of a state $\rho$ is given by
\begin{equation}\label{eq:purity}
  \gamma(\rho)=\tr(\rho^2)\,
\end{equation}
where $2^{-n}\le \gamma\le 1$, $\gamma(\rho)=1$ for any pure state $\rho=\ket{\psi}\bra{\psi}$, and for maximally mixed state $\rho_\text{m}=I_n 2^{-n}$ we have $\gamma(\rho_\text{m})=2^{-n}$. 
For pure Haar random states we find trivially ${\mathbb{E}}_{\ket{\phi}\leftarrow S(\mathcal{H})}[\gamma(\ket{\phi})]=1$.
We can efficiently measure $\gamma(\rho)$ via the SWAP test acting on $\rho\otimes\rho$, which accepts with probability $\text{Pr}_\text{SWAP}(\rho)=\frac{1}{2}+\gamma(\rho)/2$ (see Fig.~\ref{fig:SWAPtest} or~\cite{barenco1997stabilization}). This follows from $\tr(\rho^2)=\tr(S \rho \otimes \rho)$, where $S$ is the SWAP operator $S=\sum_{k\ell}\ket{k\ell}\bra{\ell k}$.

\begin{figure*}[htbp]
	\centering	
	\subfigimg[width=0.4\textwidth]{}{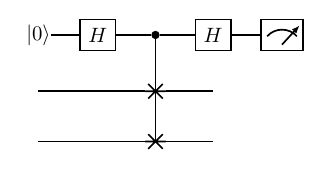}
	\caption{SWAP test.
	}
	\label{fig:SWAPtest}
\end{figure*}

Now, we regard the local depolarizing channel \begin{equation}\label{eq:depol}
\Lambda_p(\rho)=(1-p) \rho +\frac{p}{2} I_1\otimes \tr_1 (\rho)
\end{equation}
acting on the first qubit with probability $p$  where $\tr_1(.)$ is the partial trace over the first qubit. For any input pure state $\ket{\psi}$ we find $\gamma(\Lambda_p(\ket{\psi}))\le 1-p+p^2/2$.

Let us consider an ensemble of pure states $\ket{\psi}\in\mathcal{E}$. Now, we subject the states to noise, where we define the ensemble of corresponding noisy states $\rho\in \mathcal{E}_c$ where $\mathcal{E}_c=\{\Lambda_p(\ket{\psi}) : \ket{\psi}\in\mathcal{E}\}$ with $p=\Omega(n^{-c})$ with $c>0$. 
For the noisy states we have ${\mathbb{E}}_{\rho\leftarrow \mathcal{E}_c}[\gamma(\rho)]=1-\Omega(n^{-c})$ and thus
\begin{equation}
\left\vert \underset{\rho\leftarrow \mathcal{E}_c}{\mathbb{E}}[\text{Pr}_\text{SWAP}(\rho)]-\underset{\ket{\phi}\leftarrow S(\mathcal{H})}{\mathbb{E}}[\text{Pr}_\text{SWAP}(\ket{\phi})]\right\vert=\Omega(n^{-c})\,.
\end{equation}
We conclude that the SWAP test efficiently distinguishes Haar random states and any ensemble of states affected by $p=\Omega(n^{-c})$ depolarizing noise. 
Thus, PRSs can have at most $p=\text{negl}(n)$ depolarizing noise.

Next, we consider the case of noisy unitaries. 
Consider an ensemble of unitaries $U\in \mathcal{G}$. Now, the unitaries are affected by the depolarizing channel~\eqref{eq:depol}, giving us the ensemble of noisy channels $\Gamma\in \mathcal{G}_c$ with $\mathcal{G}_c=\{\Lambda_p \circ U : U\in \mathcal{G}\}$ where $p=\Omega(n^{-c})$. We now apply the noisy unitaries on a test state $\ket{0}$ and perform the SWAP test.
We find
\begin{equation}
\left\vert \underset{\Gamma\leftarrow \mathcal{G}_c}{\mathbb{E}}[\text{Pr}_\text{SWAP}(\Gamma(\ket{0})]-\underset{U\leftarrow \mathcal{U}(\mathcal{H})}{\mathbb{E}}[\text{Pr}_\text{SWAP}(U\ket{0})]\right\vert=\Omega(n^{-c})\,.
\end{equation}
The SWAP test efficiently distinguishes Haar random unitaries and any unitaries affected by $p=\Omega(n^{-c})$ depolarizing noise. 
Thus, PRUs can be subject to at most $p=\text{negl}(n)$ depolarizing noise. 
\end{proof}

\section{Proof of inefficient  testing of imaginarity for states}\label{sec:proof_imag}
Here we show that imaginarity cannot be efficiently tested for states. First, we restate the Theorem~\ref{thm:no_test_imag} of the main text:
\begin{thmS}[Imaginarity cannot be efficiently tested for states]\label{thm:no_test_imag_sup}
Any tester $\mathcal{A}_{\mathcal{I}}$ for imaginarity according to Def.~\ref{def:prop-tester} for an $n$-qubit state $\ket{\psi}$ requires $t=\Omega(2^{n/2})$ copies of $\ket{\psi}$ for $\delta<1-n^22^{-n/2}$ and any $\beta<\delta$.
\end{thmS}
\begin{proof}
Let us first recall Levy's lemma~\cite{popescu2005foundations}
\begin{lemS}[Levy's lemma]\label{def:Levy}
 Given a function $f : \mathbb{S}(d) \rightarrow\mathbb{R}$ defined on the $d$-dimensional hypersphere $\mathbb{S}(d)$, $\epsilon>0$ and a point $\phi \in \mathbb{S}(d)$ chosen uniformly at random,
\begin{equation}
    \mathrm{Pr}(\vert f(\phi)-\langle f \rangle\vert \ge \epsilon)\le 2\exp\left(\frac{-2C(d+1)\epsilon^2}{\eta^2} \right)
\end{equation}
where $\eta$ is the Lipschitz constant of f, given by $\eta = \mathrm{sup} \vert \nabla f\vert$, and $C$ is a positive constant (which can be taken
to be $C = (18\pi^3)^{-1}$.
\end{lemS}
Due to normalisation, pure states in a $2^n$-dimensional Hilbert space can be represented by points on the surface of a ($2\cdot2^n -1$)-dimensional
hypersphere $\mathbb{S}(2\cdot2^n -1)$, and hence we can apply Levy’s Lemma to functions of the randomly selected quantum state $\ket{\phi}$ by setting $d = 2\cdot2^n -1$.

Next, we bound the Lipschitz constant $\eta$ of imaginarity. For any normalized pure states $\ket{\psi}$, $\ket{\phi}$, we find
\begin{equation}\label{eq:lipschitz}
\eta =\text{max}_{\ket{\psi},\ket{\phi}} \frac{\vert \mathcal{I}(\ket{\psi})-\mathcal{I}(\ket{\phi})\vert}{\vert \ket{\psi}-\ket{\phi}\vert}=4 ,
\end{equation}
where $\vert \ket{\psi} -\ket{\phi}\vert =\sqrt{2-2\Re(\braket{\psi}{\phi})}$ indicates the Euclidean norm. 
First, we have
\begin{align*}
\vert \mathcal{I}(\ket{\psi})-\mathcal{I}(\ket{\phi})\vert&=\frac{1}{4}\left\vert\Vert  \ket{\psi}\bra{\psi}-\ket{\psi^\ast}\bra{\psi^\ast} \Vert _1^2-\Vert  \ket{\phi}\bra{\phi}-\ket{\phi^\ast}\bra{\phi^\ast} \Vert_1^2 \right\vert \\
&\le \left\vert\Vert  \ket{\psi}\bra{\psi}-\ket{\psi^\ast}\bra{\psi^\ast} \Vert _1-\Vert  \ket{\phi}\bra{\phi}-\ket{\phi^\ast}\bra{\phi^\ast} \Vert_1 \right\vert 
\end{align*}
where we used $\Vert x\Vert^2-\Vert y\Vert^2=(\vert \Vert x\Vert-\Vert y\Vert\vert)(\Vert x\Vert+\Vert y\Vert)$ and \revC{$\Vert  \ket{\psi}\bra{\psi}+\ket{\psi^\ast}\bra{\psi^\ast} \Vert _1\le 2$.}
Then, we have
\begin{align*}
\vert \mathcal{I}(\ket{\psi})-\mathcal{I}(\ket{\phi})\vert&\le
\Vert \ket{\psi}\bra{\psi}-\ket{\phi}\bra{\phi} -( \ket{\psi^\ast}\bra{\psi^\ast}-\ket{\phi^\ast}\bra{\phi^\ast}) \Vert_1 \\
&\le \Vert \ket{\psi}\bra{\psi}-\ket{\phi}\bra{\phi}\Vert_1 +\Vert \ket{\psi^\ast}\bra{\psi^\ast}-\ket{\phi^\ast}\bra{\phi^\ast} \Vert_1\\
&=4\sqrt{1-\vert \braket{\psi}{\phi}\vert^2}\le  4\vert \ket{\psi}-\ket{\phi}\vert\numberthis\label{eq:normI}
\end{align*}
where in the first two steps we used the triangle inequalities $\vert\Vert x\Vert-\Vert y\Vert \vert \le \Vert x-y\Vert$, \revC{$\Vert x-y\Vert \le \Vert x\Vert +\Vert y\Vert$}, and in the last inequality 
\begin{align*}
1-\vert \braket{\psi}{\phi}\vert^2&= 1-\Re(\braket{\psi}{\phi})^2-\Im(\braket{\psi}{\phi})^2
\le 1-\Re(\braket{\psi}{\phi})^2= (1-\Re(\braket{\psi}{\phi}))(1+\Re(\braket{\psi}{\phi})) \\
&\le2(1-\Re(\braket{\psi}{\phi}))=\vert \ket{\psi}-\ket{\phi} \vert^2\,,
\numberthis\label{eq:norms}
\end{align*}
which finally gives us $\eta=4$.

From Levy's lemma and the imaginarity of Haar random states $\mathcal{I}_\text{Haar}=\mathbb{E}_{\ket{\psi}\leftarrow S(\mathcal{H})}[\mathcal{I}(\ket{\psi})]=1-2(2^n+1)^{-1}$ (see \eqref{eq:random_imaginarity}), we have
\begin{equation}
    \mathrm{Pr}(\vert \mathcal{I}(\ket{\psi})-\mathcal{I}_\text{Haar}\vert \ge \epsilon)\le 2\exp(-\beta_0 2^n\epsilon^2 ),
\end{equation}
where $\beta_0=4C/\eta^2=(288\pi^{3})^{-1}$ and $\epsilon>0$.
We now drop the absolute value $\vert . \vert$ and define $\epsilon=\mathcal{I}_\text{Haar}-\delta$ to get
\begin{equation}
    \mathrm{Pr}(\mathcal{I}(\ket{\psi}) \le  \delta)\le 2\exp(-\beta_0 2^n(\mathcal{I}_\text{Haar}-\delta)^2),
\end{equation}
which is valid for any $\delta<\mathcal{I}_\text{Haar}$. %

\revC{We now choose $\delta$ such that we achieve exponential concentration, i.e. exponentially many quantum states have an imaginarity $\mathcal{I}>\delta$. 
We achieve this by setting 
\begin{equation}
    \delta< \mathcal{I}_\text{Haar}(n)-\beta_0^{-1/2}n2^{-n/2}=1-\frac{2}{2^n+1}-\sqrt{288\pi^3\ln(2)}n2^{-n/2}
\end{equation}
where we get
\begin{align}\label{eq:small-distance-low-probability}
P_\delta=\underset{\ket{\psi}\leftarrow S(\mathcal{H})}{\text{Pr}}\left[\mathcal{I}(\ket{\psi})< \delta\right]=O(2^{-n^2})\,.
\end{align}
We now proceed with the choice 
\begin{equation}\label{eq:choicedelta}
    \delta<1-n^22^{-n/2}\,.
\end{equation}
which satisfies the scaling of~\eqref{eq:small-distance-low-probability} for any $n\geq80$.
}

Let us now define the ensemble $\mathcal{E}_{\mathcal{I}\ge \delta}=\{\ket{\psi}:\mathcal{I}(\ket{\psi})\ge\delta, \ket{\psi}\in S(\mathcal{H})\}$ 
of states with imaginarity at least $\delta$, and the ensemble $\mathcal{E}_{\mathcal{I}< \delta}=\{\ket{\psi}:\mathcal{I}(\ket{\psi})<\delta, \ket{\psi}\in S(\mathcal{H})\}$ of states with imaginarity at most $\delta$.

Given the concentration~\eqref{eq:small-distance-low-probability} and choosing $\delta$ according to~\eqref{eq:choicedelta}, the ensemble of states with $\mathcal{I}\ge \delta$ is exponentially close to the ensemble of Haar random states $S(\mathcal{H})$ for any $t$
\begin{align}
\text{TD}\left(\underset{\ket{\psi}\leftarrow\mathcal{E}_{\mathcal{I}\ge \delta}}{\mathbb{E}}\big[\ket{\psi}\bra{\psi}^{\otimes t}\big],\underset{\ket{\phi}\leftarrow S(\mathcal{H})}{\mathbb{E}}\big[\ket{\phi}\bra{\phi}^{\otimes t}\big]\right)=O(2^{-n^2})\,,\label{eq:TDHaarImag}
\end{align}
where we have the trace distance $\text{TD}(\rho,\sigma)=\frac{1}{2}\| \rho-\sigma \|_1$.

We now prove~\eqref{eq:TDHaarImag} in the following:
\begin{align*}
&\text{TD}\left(\underset{\ket{\psi}\leftarrow\mathcal{E}_{\mathcal{I}\ge \delta}}{\mathbb{E}}\big[\ket{\psi}\bra{\psi}^{\otimes t}\big],\underset{\ket{\phi}\leftarrow S(\mathcal{H})}{\mathbb{E}}\big[\ket{\phi}\bra{\phi}^{\otimes t}\big]\right)\numberthis\label{eq:closeproofimag}\\
&=\text{TD}\left((1-P_\delta)\underset{\ket{\psi}\leftarrow\mathcal{E}_{\mathcal{I}\ge \delta}}{\mathbb{E}}\big[\ket{\psi}\bra{\psi}^{\otimes t}\big]+P_\delta\underset{\ket{\psi}\leftarrow\mathcal{E}_{\mathcal{I}< \delta}}{\mathbb{E}}\big[\ket{\psi}\bra{\psi}^{\otimes t}\big]\right.\\
&\left.\quad -P_\delta(\underset{\ket{\psi}\leftarrow\mathcal{E}_{\mathcal{I}< \delta}}{\mathbb{E}}\big[\ket{\psi}\bra{\psi}^{\otimes t}\big]-\underset{\ket{\psi}\leftarrow\mathcal{E}_{\mathcal{I}\ge \delta}}{\mathbb{E}}\big[\ket{\psi}\bra{\psi}^{\otimes t}\big]),\underset{\ket{\phi}\leftarrow S(\mathcal{H})}{\mathbb{E}}\big[\ket{\phi}\bra{\phi}^{\otimes t}\big]\right)\\
&= \text{TD}\left(\underset{\ket{\phi}\leftarrow S(\mathcal{H})}{\mathbb{E}}\big[\ket{\phi}\bra{\phi}^{\otimes t}\big]-P_\delta(\underset{\ket{\psi}\leftarrow\mathcal{E}_{\mathcal{I}< \delta}}{\mathbb{E}}\big[\ket{\psi}\bra{\psi}^{\otimes t}\big]-\underset{\ket{\psi}\leftarrow\mathcal{E}_{\mathcal{I}\ge \delta}}{\mathbb{E}}\big[\ket{\psi}\bra{\psi}^{\otimes t}\big]),\underset{\ket{\phi}\leftarrow S(\mathcal{H})}{\mathbb{E}}\big[\ket{\phi}\bra{\phi}^{\otimes t}\big]\right)\\
&= P_\delta\text{TD}\left(\underset{\ket{\psi}\leftarrow\mathcal{E}_{\mathcal{I}< \delta}}{\mathbb{E}}\big[\ket{\psi}\bra{\psi}^{\otimes t}\big],\underset{\ket{\psi}\leftarrow\mathcal{E}_{\mathcal{I}\ge \delta}}{\mathbb{E}}\big[\ket{\psi}\bra{\psi}^{\otimes t}\big]\right)=O(2^{-n^2}) ,
\end{align*}
where in the second step we used 
\begin{equation}
    (1-P_\delta)\underset{\ket{\psi}\leftarrow\mathcal{E}_{\mathcal{I}\ge \delta}}{\mathbb{E}}\big[\ket{\psi}\bra{\psi}^{\otimes t}\big]+P_\delta\underset{\ket{\psi}\leftarrow\mathcal{E}_{\mathcal{I}< \delta}}{\mathbb{E}}\big[\ket{\psi}\bra{\psi}^{\otimes t}\big]=\underset{\ket{\phi}\leftarrow S(\mathcal{H})}{\mathbb{E}}\big[\ket{\phi}\bra{\phi}^{\otimes t}\big]
\end{equation}
and in the last step $\text{TD}(\rho,\sigma)\le1$ for any valid quantum states $\rho$, $\sigma$.

\revC{Note that~\eqref{eq:TDHaarImag} is independent of the number of copies $t$. This is because the set of highly imaginary states and the set of Haar random states are nearly identical, differing only by some exponentially small fraction. Therefore, they remain close for any $t$. Note that this is in contrast to the distance between Haar random states and discrete sets of states (e.g. the $K$-subset phase states~\eqref{eq:phaseKstates}), which due to their small size become further in distance with more copies $t$.}

Next, we recall the $K$-subset phase states from~\eqref{eq:phaseKstates} or Ref.~\cite{bouland2022quantum} 
\begin{equation}\label{eq:phaseKstates_sup}
   \ket{\psi_{S,f}}=\frac{1}{\sqrt{K}}\sum_{x\in S}(-1)^{f(x)}\ket{x} , 
\end{equation}
where $S$ is a set of binary strings $\{0,1\}^{n}$ with exactly $K=\vert S\vert$ elements and $f:\{0,1 \}^n \to \{ 0, 1\}$ a binary phase function. \revC{Note that the $K$-subset phase state is real-valued, and thus has imaginarity $\mathcal{I}(\ket{\psi_{S,f}})=0$.}
Let us also define the set of all $K$-subset phase states $\mathcal{E}_\text{p}^K=\{\ket{\psi_{S,f}}\}_{S\, \text{with}\, \vert S\vert=K,\,f}$ which contains the states $\ket{\psi_{S,f}}$ with all binary phase functions $f$ and all $S$ with $K=\vert S\vert$.

The ensemble $\mathcal{E}_\text{p}$ of phase states~\eqref{eq:phaseKstates_sup} has been shown to be statistically close to the ensemble of Haar random states (Theorem 2.1 of Ref.~\cite{bouland2022quantum})
\begin{equation}\label{eq:phaseHaar}
\text{TD}\left(\underset{\ket{\psi}\leftarrow\mathcal{E}_{\text{p}}^K}{\mathbb{E}}\big[\ket{\psi}\bra{\psi}^{\otimes t}\big],\underset{\ket{\phi}\leftarrow S(\mathcal{H})}{\mathbb{E}}\big[\ket{\phi}\bra{\phi}^{\otimes t}\big]\right)= O\left(\frac{t^2}{K}\right)\,.
\end{equation}
\revC{Now, as the Haar random ensemble is exponentially close to the ensemble of highly imaginary states, we can apply the inequality $\text{TD}(\rho,\sigma)\le \text{TD}(\rho,\theta)+\text{TD}(\theta,\sigma)$ to~\eqref{eq:TDHaarImag} and~\eqref{eq:phaseHaar} to get}
\begin{align}
\text{TD}\left(\underset{\ket{\psi}\leftarrow\mathcal{E}_{\text{p}}^K}{\mathbb{E}}\big[\ket{\psi}\bra{\psi}^{\otimes t}\big],\underset{\ket{\phi}\leftarrow\mathcal{E}_{\mathcal{I}\ge \delta}}{\mathbb{E}}\big[\ket{\phi}\bra{\phi}^{\otimes t}\big]\right)= O\left(\frac{t^2}{K}\right)\,.\label{eq:TDphaseImag}
\end{align}
\eqref{eq:TDphaseImag} implies that the ensemble of phase states and the ensemble of states with imaginarity $\mathcal{I}(\ket{\psi})\ge \delta$ (where $\delta$ given by~\eqref{eq:choicedelta}) are exponentially close in trace-distance  unless $t=\Omega(\sqrt{K})$.

Two ensembles can be distinguished only if they are sufficiently far in TD distance.
In particular, for any state discrimination protocol between two ensembles $\rho$, $\sigma$, the maximal possible discrimination probability is given by the Helstrom bound~\cite{bae2015quantum}
\begin{equation}\label{eq:helstrom}
    P_\text{discr}(\rho,\sigma)=\frac{1}{2}+\text{TD}(\rho,\sigma)\,.
\end{equation}
In particular, any algorithm trying to distinguish two ensembles with $P_\text{discr}\ge2/3$ requires at least $\text{TD}(\rho ,\sigma)\ge1/6$. 

We now consider the case $K=2^n$, i.e. subset phase states with support on all $2^n$ computational basis states.
\revC{To  achieve $P_\text{discr}=\Omega(1)$ between $t$ copies of $\mathcal{E}_\text{p}^{2^n}$ and $\mathcal{E}_{\mathcal{I}\ge \delta}$, we require  $t=\Omega(2^{n/2})$ for $\delta(n)<1-n^22^{-n/2}$. Thus, testing imaginarity must also require $t=\Omega(2^{n/2})$ copies, else imaginarity testing could be used to distinguish those two ensembles. }
\end{proof}

Note that for the proof we used the subset phase states with $\mathcal{I}=0$ which is the extremal case (i.e. $\beta=0$). Thus, our result holds also for property testing of imaginarity for any $\beta$ with $0\le\beta <\delta$.

\section{Inefficient algorithm to measure imaginarity of states}\label{sec:meas_imag_ineff}
Now, we discuss a scheme to measure imaginarity for pure states. In particular, we have
\begin{equation}
\label{eq:imaginarity_pure_states}
\mathcal{I}(\ket{\psi})=1-2^{n}\bra{\psi}^{\otimes 2}(\ket{\Phi}\bra{\Phi})\ket{\psi}^{\otimes 2}\,,
\end{equation}
where $\ket{\Phi}=2^{-n/2}\sum_{k=0}^{2^n-1}\ket{k}\ket{k}$ is the maximally entangled state. To measure imaginarity, one computes the probability of the projector $\ket{\Phi}\bra{\Phi}$ for two copies of state $\ket{\psi}^{\otimes 2}$. Using the Ricochet property~\ref{eq:ricochet}, it is easy to see for $\ket{\psi}=U\ket{0}$ with arbitrary unitary $U$ via
\begin{align*}
&\bra{\psi}^{\otimes 2}\ket{\Phi}\bra{\Phi}\ket{\psi}^{\otimes 2}=
\bra{0}^{\otimes 2}U^{\dagger}\otimes U^{\dagger}\ket{\Phi}\bra{\Phi}U\otimes U\ket{0}^{\otimes 2}\\
=&
\bra{0}^{\otimes 2}U^{\dagger}{U^\text{T}}^{\dagger}\otimes I \ket{\Phi}\bra{\Phi}U^\text{T}U\otimes I\ket{0}^{\otimes 2}\\
=&2^{-n}\bra{0}{U^{\dagger} U^\text{T}}^{\dagger}\ket{0}\bra{0}U^\text{T}U\ket{0}=2^{-n}\vert\braket{\psi}{\psi^*}\vert^2 ,
\end{align*}
where we made use of~\eqref{eq:ricochet} and ${U^\text{T}}^\dagger=U^{\ast}$. This completes the proof of \eqref{eq:imaginarity_pure_states}.

We measure the observable $2^n\ket{\Phi}\bra{\Phi}$ which has two possible outcomes with eigenvalues $\lambda_1=2^n$ and $\lambda_2=0$, occurring with probability $p=\bra{\psi}^{\otimes 2}\ket{\Phi}\bra{\Phi}\ket{\psi}^{\otimes 2}$ and $q=1-p$. We can see this as Bernoulli trials. 
The Chernoff bound for the average over $M$ Bernoulli trials $\hat{X}=\frac{1}{M}\sum_{i=1}^M X_i$ with trial outcome $X_i$ is given by
\begin{equation}
    P(\vert \hat{X} - \mu \vert \ge \tilde{\epsilon}) \le \exp(-\frac{M\tilde{\epsilon}^2}{3\mu})\,,
\end{equation}
where $\tilde{\epsilon}$ is the additive estimation error.
Now we have the estimator of imaginarity $\hat{\mathcal{I}}=1-2^n \hat{X}$ and the mean value $\mu=p=2^{-n}(1-\mathcal{I})$.
From this, we get
\begin{equation}
    P(\vert \hat{\mathcal{I}} - \mathcal{I} \vert \ge \epsilon)\equiv \delta \le \exp(-\frac{M\epsilon^2}{3(1-\mathcal{I})2^n})\,.
\end{equation}
and finally
\begin{equation}
    M\le 3\epsilon^{-2}(1-\mathcal{I})2^n \log(1/\delta)\le 3\epsilon^{-2} 2^n\log(1/\delta)\,.
\end{equation}
Thus, to achieve a given accuracy $\epsilon$ with failure probability $\delta$, we require at most a number of measurements $M=O(2^n\epsilon^{-2}\log(1/\delta)$. 
This $O(2^n)$ scaling is exponential, but  better than tomography with $O(4^n)$.

\section{Efficient measurement of imaginarity for stabilizer states}\label{sec:stabMeasimag}
\begin{claimS}[Imaginarity of stabilizer states can be measured efficiently]\label{thm:meas_imag_stab_sup}
Given a stabilizer state $\ket{\psi}\in\mathrm{STAB}$, imaginarity can be computed via
\begin{equation}
\mathcal{I}_\mathrm{STAB}(\ket{\psi})=1-2^{-n}\sum_{\sigma\in\mathcal{P}}\vert\bra{\psi}\sigma\ket{\psi^*}\vert^2\vert\bra{\psi}\sigma\ket{\psi}\vert^2 ,
\end{equation}
where $\mathcal{P}$ is the set of all Pauli strings with phase $+1$.
$\mathcal{I}_\mathrm{STAB}$ can be efficiently measured within additive precision $\delta$ with a failure probability $\nu$ using at most $t=O(\delta^{-2}\log(1/\nu))$ copies.
\end{claimS}
\begin{proof}
We denote the Pauli matrices by $\sigma_{00}=I_{2}$, $\sigma_{01}=\sigma^x$, $\sigma_{10}=\sigma^z$ and $\sigma_{11}=\sigma^y$.
We define the set of all $4^n$ Pauli strings $\mathcal{P}=\{\sigma_{\boldsymbol{r}}\}_{\boldsymbol{r}}$ with $+1$ phase by $N$-qubit tensor products of Pauli matrices given by $\sigma_{\boldsymbol{r}}=\bigotimes_{j=1}^n \sigma_{\boldsymbol{r}_{2j-1}\boldsymbol{r}_{2j}}$ with $\boldsymbol{r}\in\{0,1\}^{2n}$.

\begin{figure*}[htbp]
	\centering	
	\subfigimg[width=0.3\textwidth]{}{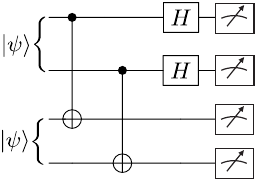}
	\caption{Bell measurement for state $\ket{\psi}$.
	}
	\label{fig:Bell}
\end{figure*}

A stabilizer state $\ket{\psi_\mathrm{STAB}}$ is defined by a commuting subgroup $G$ of $\abs{G}=2^n$ Pauli strings $\sigma$. 
We have $\bra{\psi_\mathrm{STAB}}\sigma\ket{\psi_\mathrm{STAB}}=\pm1$ for $\sigma\in G$ and  $\bra{\psi_\mathrm{STAB}}\sigma'\ket{\psi_\mathrm{STAB}}=0$ for $\sigma'\notin G$~\cite{gottesman1998heisenberg}.

\begin{algorithm}[ht]
 \SetAlgoLined
 \LinesNumbered
  \SetKwInOut{Input}{Input}
  \SetKwInOut{Output}{Output}
   \Input{  $4M$ copies of $\ket{\psi}$
   }
    \Output{$\mathcal{I}_\mathrm{STAB}$
    }

$b=0$

 \SetKwRepeat{Do}{do}{while}
    \For{$\ell=1,\dots,M$}{
    Prepare $\ket{\eta}=U_\text{Bell}^{\otimes n}\ket{\psi}\ket{\psi}$

    Sample $\boldsymbol{r}\sim \vert \braket{\boldsymbol{r}}{\eta}\vert^2$
    
     \SetKwRepeat{Do}{do}{while}
    \For{$k=1,2$}{
    Prepare $\ket{\psi}$ and measure in eigenbasis of Pauli string $\sigma_{\boldsymbol{r}}$ to get eigenvalue $\lambda_k\in\{+1,-1\}$

    }
    $b= b+\lambda_1\lambda_2$
    }
    $\mathcal{I}_\mathrm{STAB}= 1-b/M$
    
 \caption{Imaginarity of stabilizer states}
 \label{alg:imagwitness}
\end{algorithm}

We now propose Algorithm~\ref{alg:imagwitness} to efficiently measure the imaginarity of stabilizer states.
We recall the Bell transformation $U_\text{Bell}=(H\otimes I_1) \text{CNOT}$ (see also Fig.~\ref{fig:Bell}),
where $H=\frac{1}{\sqrt{2}}(\sigma^x+\sigma^z)$ is the Hadamard gate, and $\text{CNOT}=\exp(i\frac{\pi}{4}(I_1-\sigma^z)\otimes(I_1-\sigma^x))$. We apply the Bell transformation on two copies $\ket{\psi}\otimes\ket{\psi}$ and measure in the computational basis. The outcome $\boldsymbol{r}$ corresponding to computational basis state $\ket{\boldsymbol{r}}$ is sampled with a probability~\cite{montanaro2017learning}
\begin{equation}\label{eq:probBell}
P(\boldsymbol{r})=\bra{\psi}\bra{\psi}{U_\text{Bell}^{\otimes n}}^\dagger \ket{\boldsymbol{r}}\bra{\boldsymbol{r}}U_\text{Bell}^{\otimes n}\ket{\psi}\ket{\psi}=2^{-n}\vert\bra{\psi}\sigma_{\boldsymbol{r}}\ket{\psi^*}\vert^2\,.
\end{equation} 
Using a combination of Bell measurements and Pauli measurements~\cite{gross2021schur}, we propose to measure the imaginarity of stabilizer states via
\begin{equation}\label{eq:imag_stab}
\mathcal{I}_\mathrm{STAB}(\ket{\psi}):=1-
\underset{\boldsymbol{r}\sim P(\boldsymbol{r})}{\mathbb{E}}[\vert\bra{\psi}\sigma_{\boldsymbol{r}}\ket{\psi}\vert^2].
\end{equation}
In particular, given a stabilizer state $\ket{\psi}\in \mathrm{STAB}$ we find
\begin{align*}
\mathcal{I}_\mathrm{STAB}(\ket{\psi})&=1-2^{-n}\sum_{\sigma\in\mathcal{P}}\vert\bra{\psi}\sigma\ket{\psi^*}\vert^2\vert\bra{\psi}\sigma\ket{\psi}\vert^2\\
&=1-2^{-n}\sum_{\sigma\in G}\vert\bra{\psi}\sigma\ket{\psi^*}\vert^2=1-\vert\braket{\psi}{\psi^*}\vert^2\equiv\mathcal{I}(\ket{\psi})\, ,
\end{align*}
where in the second line we used $\sigma\ket{\psi}=\pm\ket{\psi}$ for $\ket{\psi}\in \mathrm{STAB}$ with corresponding $\sigma\in G(\ket{\psi})$, and $\bra{\psi}\sigma\ket{\psi}=0$ for $\sigma \notin G(\ket{\psi})$.

Note that non-real stabilizer states $\ket{\psi_\mathrm{STAB}^I}\notin\mathbb{R}^{2^n}$ are always maximally imaginary, i.e. $\mathcal{I}(\ket{\psi_\mathrm{STAB}^I})=1$ and $\vert\braket{\psi_\mathrm{STAB}^I}{{\psi^I_\mathrm{STAB}}^*}\vert^2=0$. This can be easily seen from the fact  that $\vert\bra{\psi_\mathrm{STAB}}\sigma\ket{\psi_\mathrm{STAB}}\vert^2\in\{0,1\}$ and $\ket{\psi^*_\mathrm{STAB}}=\bigotimes_{k=1}^{n}\sigma_z^{\alpha_k}\ket{\psi_\mathrm{STAB}}$  with some $\alpha_k\in\{0,1\}$~\cite{montanaro2017learning}.  Thus, it follows $\vert\braket{\psi_\mathrm{STAB}}{\psi^*_\mathrm{STAB}}\vert^2=\vert\bra{\psi_\mathrm{STAB}}\bigotimes_{k=1}^{n}\sigma_z^{\alpha_k}\ket{\psi_\mathrm{STAB}}\vert^2$. For any imaginary stabilizer, we necessarily have $\vert\braket{\psi_\mathrm{STAB}^I}{{\psi^I_\mathrm{STAB}}^*}\vert^2<1$, and therefore $\vert\braket{\psi_\mathrm{STAB}^I}{{\psi^I_\mathrm{STAB}}^*}\vert^2=0$.

\revC{Each repetition of the algorithm yields the outcome $\lambda_1=-1$ or $\lambda_2=1$ with a range of outcomes $\Delta \lambda=2$.
We can upper bound the number of repetitions $M$ needed to estimate $\mathcal{I}_\text{STAB}$ using Hoeffding's inequality:
\begin{equation}\label{eq:P_I}
    \text{Pr}(\vert \hat{\mathcal{I}}_\text{STAB} -\mathcal{I}_\text{STAB}\vert\ge \delta)\equiv\nu\le 2\exp\left(-\frac{2\delta^2 M}{\Delta\lambda^2}\right)\,,
\end{equation}
where $\hat{\mathcal{I}}_\text{STAB}$ is the estimation of exact value $\mathcal{I}_\text{STAB}$ from $M$ repetitions, $\delta$ is the allowed error, $\Delta\lambda=1$ is the range of possible measurement outcomes and $\nu$ is the failure probability of getting an error larger than $\delta$. 
By inverting we find
\begin{equation}\label{eq:hoeffding}
    M\le \frac{\Delta\lambda^2}{2\delta^2}\log\left(\frac{2}{\nu}\right)=\frac{1}{\delta^2}\log\left(\frac{2}{\nu}\right)\,,
\end{equation}
To achieve an error of at most $\delta$ with a probability of failure $\nu$, we require at most $M=O(\delta^{-2}\log(\nu^{-1}))$ repetitions in Algorithm~\ref{alg:imagwitness} and $4M$ copies of $\ket{\psi}$. The classical post-processing time scales as $O(n)$, which is asymptotically optimal.}
As real stabilizer states have $\mathcal{I}=0$, while imaginary stabilizer states must have $\mathcal{I}=1$, this directly implies that one can efficiently distinguish real and complex-valued stabilizer states with $t=O(1)$ copies of $\ket{\psi}$. 
\end{proof}

We also propose~\eqref{eq:imag_stab} as an efficient witness of non-imaginarity, which we define in analogy to entanglement witnesses~\cite{guhne2009entanglement}. When a non-imaginarity witness is zero, the state must be real. When the witness is greater than zero, then the state may be either real or imaginary.
~\eqref{eq:imag_stab} fulfills this condition as it is zero only when for real stabilizer states, and greater than zero for all other states~\cite{gross2021schur}. 

\section{Efficient measurement of imaginarity for unitaries}\label{sec:meas_imag_U}
Here, we show that measuring imaginarity $\mathcal{I}_\text{p}(U)$ of an $n$-qubit unitary $U$ is efficient. Imaginarity for unitaries can be defined via its Choi-state
\begin{align*}
\mathcal{I}_\text{p}(U) &\equiv  \mathcal{I}(U\otimes I \ket{\Phi})=1-\vert \bra{\Phi}(U^\dagger \otimes I)  (U^*\otimes I) \ket{\Phi}\vert^2= 1-4^{-n}\left\vert\tr(U^\dagger U^\ast)\right\vert^2\,.
\end{align*}
Let us restate the Claim to be proven:
\begin{claimS}[Imaginarity of unitaries can be measured efficiently]\label{clm:meas_unitary_sup}
Measuring $\mathcal{I}_\text{p}(U)$ within additive precision $\delta$ with a failure probability $\nu$ requires at most $t=O(\delta^{-2}\log(1/\nu))$ queries to $U$.
\end{claimS}
\begin{proof}
First, we define the maximally entangled state $\ket{\Phi}=2^{-n/2}\sum_{k=0}^{2^n-1}\ket{k}\ket{k}$ over $2n$ qubits.
Then, we prepare the state $\ket{\eta}=U\otimes U\ket{\Phi}$. 
We then measure the expectation value $P_{\Phi}(U)=\vert\braket{\Phi}{\eta}\vert^2$ of the projector $\ket{\Phi}\bra{\Phi}$
\begin{align*}
P_\Phi(U)&=4^{-n}\Big\vert\sum_{q,k} \bra{q}\bra{q}U\otimes U\ket{k}\ket{k}\Big\vert^2= 4^{-n}\Big\vert\sum_{q,k} \bra{q}\bra{q}U U^T\otimes I\ket{k}\ket{k}\Big\vert^2= 4^{-n}\Big\vert\sum_{k} \bra{k}U U^T\ket{k}\Big\vert^2\\
&=4^{-n}\left\vert\tr(U U^T)\right\vert^2=4^{-n}\left\vert\tr(U^\dagger U^\ast)\right\vert^2 ,
\end{align*}
where we used the Ricochet property~\cite{khatri2019quantum}
\begin{equation}\label{eq:ricochet}
\sum_{k=0}^{2^n-1}A\otimes I_n\ket{k}\ket{k}=\sum_{k=0}^{2^n-1}I_n\otimes A^T\ket{k}\ket{k}\,
\end{equation}
for any $n$-qubit operator $A$.
Thus, we can write imaginarity as
\begin{equation}
    \mathcal{I}_\text{p}(U)=1-P_\Phi(U)\,.
\end{equation}
\revC{$P_\Phi(U)$ can be measured efficiently by sampling in the Bell basis and measuring the probability of sampling $\ket{\Phi}$. We upper bound the required number of measurements $M$ with Hoeffding's inequality~\eqref{eq:hoeffding}, where we have $\Delta\lambda=1$ as the range of possible measurement outcomes.}
To achieve an error of at most $\delta$ with a probability of failure $\nu$, we require at most $M=O(\delta^{-2}\log(\nu^{-1}))$ measurements and $2M$ instances of $U$. The classical post-processing time scales as $O(n)$, which is asymptotically optimal.
\end{proof}

\section{PRUs require imaginarity}\label{sec:imag_PRU}
\begin{thmS}\label{thm:imag_unitary_sup}
    PRUs require imaginarity $\mathcal{I}_\text{p}=1-\mathrm{negl}(n)$. 
\end{thmS}
\begin{proof}
We measure $P_{\ket{\Phi}}=1-\mathcal{I}_\text{p}(U)$ which was introduced in the proof for Claim~\ref{clm:meas_unitary_sup}. 
We can interpret the measurement as a tester: We prepare  $U\otimes U\ket{\Phi}$ and measure the projection onto $\ket{\Phi}\bra{\Phi}$. The tester accepts when we get $\ket{\Phi}$ as outcome, else the tester rejects. 

First, we consider a set of $n$-qubit unitaries $\mathcal{E}_c$ with imaginarity $\mathbb{E}_{U\leftarrow \mathcal{E}_c}[\mathcal{I}_\text{p}(U)]=1-\Omega(n^{-c})$. For the acceptance probability of the test, we find $\mathbb{E}_{U\leftarrow \mathcal{E}_c} [P_\Phi(U)]=\Omega(n^{-c})$. 

In contrast, for Haar random unitaries we find
$\mathbb{E}_{U\leftarrow \mathcal{U}(\mathcal{H})}[P_{\Phi}(U)]=2(2^{n}(2^{n}+1))^{-1}=O(4^{-n})$~(\SM~\ref{sec:imaginarity_Haar_unitaries}). 

Thus, one can efficiently distinguish Haar random unitaries from $\mathcal{E}_c$, which places a lower bound the imaginarity of PRUs as $\mathcal{I}_\text{p}=1-\mathrm{negl}(n)$.
\end{proof}

\section{Coherence testing for states}\label{sec:coherenceTesting}

We now show the limits on coherence testing for states. 
First, we recall the definition of relative entropy of coherence for pure states $\ket{\psi}$
\begin{equation}
    \mathcal{C}(\ket{\psi})=-\sum_k{\vert c_k\vert^2}\ln(\vert c_k\vert^2).
\end{equation}
Next, we restate the theorem of the main text:
\begin{thmS}[Lower bound on coherence testing for states]\label{eq:no_test_coh_sup}
Any tester $\mathcal{A}_{\mathcal{C}}$ for relative entropy of coherence according to Def.~\ref{def:prop-tester} for an $n$-qubit state $\ket{\psi}$ requires $t=\Omega(2^{\beta/2})$ for $\delta<n\ln(2)-1$ and any $\beta<\delta$.
\end{thmS}
\begin{proof}
First, we note that Haar random states concentrate around a large, but non-maximal value of coherence~\cite{singh2016average}.
The coherence of $n$-qubit Haar random states is on average
\begin{equation}
\mathcal{C}_\text{Haar}=\underset{\ket{\phi}\leftarrow S(\mathcal{H})}{\mathbb{E}}[\mathcal{C}(\ket{\phi})]=\sum_{k=2}^{2^n}\frac{1}{k}\,,
\end{equation}
where for $n\gg1$ we have $\mathcal{C}_\text{Haar}= \gamma_\text{EM} -1 +n\ln(2)+O(2^{-n})$, where $\gamma_\text{EM}=0.57721$ is the Euler–Mascheroni constant.
Haar random states concentrate around their average value with exponentially high probability~\cite{singh2016average}
\begin{equation}
    \mathrm{Pr}(\vert \mathcal{C}(\ket{\psi})-\mathcal{C}_\text{Haar}\vert \ge \epsilon)\le 2\exp\left(-\frac{2^n\epsilon^2}{36\pi^3\ln(2)(n\ln 2)^2} \right)\,.
\end{equation}
\revC{Dropping the absolute value and defining $\epsilon=\mathcal{C}_\text{Haar}-\delta$, $\delta<\mathcal{C}_\text{Haar}$, we find that
\begin{equation}
    \mathrm{Pr}( \mathcal{C}(\ket{\psi}) \leq \delta)\le 2\exp\left(-\frac{2^n(\mathcal{C}_\text{Haar}-\delta)^2}{36\pi^3\ln(2)(n\ln 2)^2} \right)\, .
\end{equation}
When we choose $\delta<\mathcal{C}_\text{Haar}-n^2 2^{-n/2}\sqrt{36\pi^3\ln(2)^3}$, we find exponential concentration, i.e. most states have a coherence $\mathcal{C}>\delta$
\begin{equation}\label{eq:coherence_concentrate}
    \mathrm{Pr}(\mathcal{C}(\ket{\psi}) < \delta)\le O(2^{-n^2})\,.
\end{equation}
We relax this condition to $\delta< n\ln(2)-1$, which is valid for any $n\geq 30$.}

We now consider the ensemble $\mathcal{E}_{\mathcal{C}\ge \delta}$ of all states with a relative entropy of coherence greater $\delta$.
Given the concentration~\eqref{eq:coherence_concentrate}, the ensemble $\mathcal{E}_{\mathcal{C}\ge \delta}$ is exponentially close to the ensemble $S(\mathcal{H})$ of random states for any $t$
\begin{align}
\text{TD}\left(\underset{\ket{\psi}\leftarrow\mathcal{E}_{\mathcal{C}\ge \delta}}{\mathbb{E}}\big[\ket{\psi}\bra{\psi}^{\otimes t}\big],\underset{\ket{\phi}\leftarrow S(\mathcal{H})}{\mathbb{E}}\big[\ket{\phi}\bra{\phi}^{\otimes t}\big]\right)=O(2^{-n^2})\,,\label{eq:TDHaarCoh}
\end{align}
where the proof follows the one for imaginarity given in~\eqref{eq:closeproofimag}.
\revC{Note that the distance between the states of high coherence and Haar random states is independent of the number of copies $t$. This is because of the concentration~\eqref{eq:coherence_concentrate}, i.e the two sets of states are very close to each other, and the size of the sets differ by only a small fraction. }

Further, we recall that Haar random states are close to $K$-subset phase states~\eqref{eq:phaseHaar}.
We recall the triangle inequality $\text{TD}(\rho,\sigma)\le \text{TD}(\rho,\theta)+\text{TD}(\theta,\sigma)$ and apply it on~\eqref{eq:TDHaarCoh} and~\eqref{eq:phaseHaar} to get for any $\delta \le n\ln(2)-1$
\begin{align}
\text{TD}\left(\underset{\ket{\psi}\leftarrow\mathcal{E}_\text{p}^K}{\mathbb{E}}\big[\ket{\psi}\bra{\psi}^{\otimes t}\big],\underset{\ket{\phi}\leftarrow\mathcal{E}_{\mathcal{C}\ge \delta}}{\mathbb{E}}\big[\ket{\phi}\bra{\phi}^{\otimes t}\big]\right)\leq O\left(\frac{t^2}{K}\right).\label{eq:TDphaseCoh}
\end{align}
With~\eqref{eq:TDphaseCoh} and Helstrom bound~\eqref{eq:helstrom}, we require $t=\Omega(\sqrt{K})$ copies to distinguish $K$-subset phase states from states drawn from the ensemble of states with $\mathcal{C}\ge\delta$ where $\delta >n\ln(n)-1$.

Note that for any $K$-subset phase state $\ket{\psi}\in\mathcal{E}_{\text{p}}^K$, the coherence is given by $\mathcal{C}(\ket{\psi})=\ln(K)$.

The complexity of distinguishing $K$-subset phase state (which have coherence $\mathcal{C}=\log(K)$) from states with coherence $\mathcal{C}>\delta$ gives us directly  a lower bound on coherence testing.
In particular, to test whether an  arbitrary state has coherence $\mathcal{C}=\beta$ or coherence $\mathcal{C}>\delta$ with $\beta<\delta$, one requires $t=\Omega(2^{\beta/2})$ copies.

\end{proof}

Next, we show that one can efficiently measure coherence as defined by the Hilbert-Schmidt norm~\cite{baumgratz2014quantifying}
\begin{equation}
    \mathcal{C}_2(\ket{\psi})=1-\sum_k{\vert c_k\vert^4}\,,
\end{equation}
where $\ket{\psi}=\sum_k c_k \ket{k}$ with amplitudes $c_k$~\cite{baumgratz2014quantifying}. 
We have $0\le\mathcal{C}_2\le 1-2^{-n}$. Further,  $\mathcal{C}_2(\ket{\psi})=0$ if and only if $\mathcal{C}(\ket{\psi})=0$.

For this measure of coherence, we have an efficient measurement protocol:
\begin{propositionS}[Efficient measurement of coherence]\label{prop:test_coh_low_sup}
Measuring $\mathcal{C}_2(\ket{\psi})$ within additive precision $\delta$ with a failure probability $\nu$ requires at most $t=O(\delta^{-2}\log(1/\nu))$ copies.
\end{propositionS}
\begin{proof}
$\mathcal{C}_2$ is measured efficiently with the projector $\Pi_\mathcal{C}= \bigotimes_{k=1}^n(\ket{00}\bra{00}+\ket{11}\bra{11})$
\begin{equation}\label{eq:coh_tester}
    \mathcal{C}_2(\ket{\psi})=1-\bra{\psi}^{\otimes 2}\Pi_\mathcal{C} \ket{\psi}^{\otimes 2}\,.
\end{equation}
The projector $\Pi_{\mathcal{C}}$ has possible eigenvalues $\lambda=0$ or $\lambda=1$, with eigenvalue range $\Delta \lambda=1$.
Using Hoeffding's inequality~\eqref{eq:hoeffding}, we can upper bound the number of copies $M=O(1)$ for reaching a fixed precision. 
\end{proof}

Now, as stated in the main text, we prove that PRSs require sufficient relative entropy of coherence:
\begin{claimS}\label{clm:pseudo_coh_sup}
PRSs require $\mathcal{C}=\omega(\log(n))$ coherence.
\end{claimS}
\begin{proof}
We show that one can efficiently distinguish Haar random states and states with $\mathcal{C}=O(\log(n))$.

We now use the following tester: We measure the projector $\Pi_\mathcal{C}$ on the state $\ket{\psi}^{\otimes 2}$ from~\eqref{eq:coh_tester}. We say the tester accepts when the projector $\Pi_\mathcal{C}$ is successful, else we say the tester rejects the state.

Lets regard an ensemble $\mathcal{E}_c$ of states which have on average $\mathbb{E}_{\ket{\psi}\leftarrow\mathcal{E}_c}[\mathcal{C}_2(\ket{\psi})]=1-\Omega(n^{-c})$.
We have acceptance probability $\mathbb{E}_{\ket{\psi}\leftarrow\mathcal{E}_c}[\bra{\psi}^{\otimes 2}\Pi_\mathcal{C} \ket{\psi}^{\otimes 2}]=\Omega(n^{-c})$.

For Haar random states, we can calculate the average coherence by standard Haar integration
\begin{align*}
\mathbb{E}_{\ket{\phi}\leftarrow S(\mathcal{H})}[\mathcal{C}_2(\ket{\phi})] &=\int_U \mathcal{C}_2(U\ket{0})\text{d}U= 
1-\sum_k\int_U  \vert \bra{k}U\ket{0}\vert^4\text{d}U=1-2^{n}\int_U \vert \bra{0}U\ket{0}\vert^4\text{d}U\\
&=1-2^{n}\int_U  \tr(U^\dagger \ket{0}\bra{0}U\ket{0}\bra{0}U^\dagger \ket{0}\bra{0}U\ket{0}\bra{0})\text{d}U=
1-\frac{2}{2^{n}+1}
\end{align*}
where at the third equality we used the left invariance of Haar random unitaries and in the last step a standard identity for Haar random integrals.
We thus find $\mathbb{E}_{\ket{\phi}\leftarrow S(\mathcal{H})}[\mathcal{C}_2(\ket{\phi})]=1-O(2^{-n})$ and thus $\mathbb{E}_{\ket{\psi}\leftarrow S(\mathcal{H})}[\bra{\psi}^{\otimes 2}\Pi_\mathcal{C} \ket{\psi}^{\otimes 2}]=O(2^{-n})$. 

By repeating the testing a polynomial number of times, we can distinguish Haar random states and states drawn from $\mathcal{E}_c$ with $\Omega(1)$ probability.

As we can efficiently distinguish Haar random states and states with $\mathcal{C}_2=1-\Omega(n^{-c})$, PRSs must have $\mathcal{C}_2=1-2^{-\omega(\log(n))}$. 
From Jensen's inequality, we have $\mathcal{C}(\ket{\psi})\ge -\ln(1-\mathcal{C}_2(\ket{\psi}))$. Thus, PRSs have $\mathcal{C}(\ket{\psi})=\omega(\log(n))$.

\end{proof}

\section{PRUs and coherence power}\label{sec:coh_unitary}
We now show that PRUs require a sufficient amount of coherence. First, we restate the theorem from the main text:
\begin{thmS}\label{thm:pseudo_coh_unitary_sup}
PRUs require $\mathcal{C}_\text{p}=\omega(\log(n))$ relative entropy of coherence power.
\end{thmS}
\begin{proof}
Let us define the following efficient testing algorithm for coherence power of an $n$-qubit unitary $U$: We efficiently prepare the $2n$-qubit state 
\begin{equation}
    \chi(U)=(\mathcal{D}_n U \mathcal{D}_n U)^{\otimes 2}(\ket{\Phi}\bra{\Phi}),
\end{equation} 
where $\mathcal{D}_n(\rho)=\sum_{k=0}^{2^{n}-1}\ket{k}\bra{k}\bra{k}\rho\ket{k}$ is the $n$-qubit coherence-destroying channel. $\mathcal{D}_n$ can be efficiently implemented by measuring the state $\rho$ in the computational basis.
Next, we perform the SWAP test between the bipartition of $\chi(U)$. The SWAP test is shown in Fig.~\ref{fig:SWAPtest}. We say the SWAP test accepts when the measurement returns $0$ which occurs with probability $\text{Pr}_\text{SWAP}(U)=(1+\text{tr}(S\chi(U)))/2$ where $S=\sum_{k\ell}=\ket{k\ell}\bra{\ell k}$ is the SWAP operator. 
This measurement gives us the coherence power via~\cite{zanardi2017coherence}
\begin{equation}
    \text{tr}(S\chi(U))=2^{-n}\sum_{k,\ell} \vert \bra{k}U\ket{\ell} \vert ^4\,.
\end{equation}
Let us now consider a set of unitaries $\mathcal{E}_c$ with on average $\mathbb{E}_{U \leftarrow \mathcal{E}_c}[2^{-n}\sum_{k,\ell} \vert \bra{k}U\ket{\ell} \vert ^4]=\Omega(n^{-c})$ with $c>0$.  The SWAP test accepts with probability $\mathbb{E}_{U \leftarrow \mathcal{E}_c}[\text{Pr}_\text{SWAP}(U)]=\frac{1}{2}+\Omega(n^{-c})$.
For Haar random states, we have $\mathbb{E}_{U\leftarrow \mathcal{U}(\mathcal{H})}[2^{-n}\sum_{k,\ell} \vert \bra{k}U\ket{\ell} \vert ^4]=2^{-n+1}(1+2^{-n})^{-1}$~\cite{zanardi2017coherence} and thus
$\mathbb{E}_{U\leftarrow \mathcal{U}(\mathcal{H})}[\text{Pr}_\text{SWAP}(U)]=\frac{1}{2}+O(2^{-n})$. By repeating the test a polynomial number of times, we can distinguish Haar random unitaries and states drawn from $\mathcal{E}_c$ with $\Omega(1)$ probability.

Thus, an ensemble of PRUs $\mathcal{E}_\text{PR}$ must have $\mathbb{E}_{U \leftarrow {\mathcal{E}_\text{PR}}}[2^{-n}\sum_{k,\ell} \vert \bra{k}U\ket{\ell} \vert ^4]=2^{-\omega(\log(n))}$. 
From Jensen's inequality, we have $\mathbb{E}_{U \leftarrow {\mathcal{E}_\text{PR}}}[-\ln(2^{-n}\sum_{k,\ell} \vert \bra{k}U\ket{\ell} \vert ^4)]=\omega(\log(n))$.

Now, note that we have $2^{-n}\sum_{k,\ell} \vert \bra{k}U\ket{\ell} \vert ^2=1$, where $\xi_{k\ell}(U)=2^{-n}\vert \bra{k}U\ket{\ell} \vert ^2$ can be thought of as a probability distribution.
From the well-known inequalities of R\'enyi entropies, we have
\begin{equation}\label{eq:ineuq_ren}
    \ln(2^{-n}\text{card}(U))\ge -2^{-n}\sum_{k,\ell} \vert \bra{k}U\ket{\ell} \vert ^2\ln(\vert \bra{k}U\ket{\ell} \vert ^2) \ge -\ln(2^{-n}\sum_{k,\ell} \vert \bra{k}U\ket{\ell} \vert ^4)\,.
\end{equation}
Finally, we use $-2^{-n}\sum_{k,\ell} \vert \bra{k}U\ket{\ell} \vert ^2\ln(\vert \bra{k}U\ket{\ell} \vert ^2)\ge-\ln(2^{-n}\sum_{k,\ell} \vert \bra{k}U\ket{\ell} \vert ^4)$
to get $\mathbb{E}_{U \leftarrow {\mathcal{E}_\text{PR}}}[\mathcal{C}(U)]=\omega(\log(n))$.
\end{proof}

\section{Imaginarity for rebits}\label{sec:rebits}
While rebits are by definition real, we can define imaginarity such that it matches the definition of qubits.

We have the $n$-qubit qubit state
\begin{equation}\ket{\psi}=\sum_{k=0}^{2^n-1}(a_k+ib_k)\ket{k}
\end{equation}
and the corresponding $(n+1)$-rebit state \begin{equation}
\ket{\psi_\text{R}}=\sum_{k=0}^{2^n-1}( a_k\ket{k}\ket{R}+b_k\ket{k}\ket{I})
\end{equation}
with coefficients $a_k,b_k\in\mathbb{R}$. For normalization, we have 
\begin{equation}
\vert \braket{\psi}{\psi}\vert^2=\vert \braket{\psi_\text{R}}{\psi_\text{R}}\vert^2=\sum_k ( a_k^2+ b_k^2)=1\,.
    \end{equation}
First, we note 
\begin{equation}
    \vert \braket{\psi}{\psi^\ast}\vert^2=\left(\sum_k \left(a_k^2- b_k^2\right)\right)^2+4 \left(\sum_k a_k b_k\right)^2=\left(\sum_k a_k^2\right)^2+\left(\sum_k b_k^2\right)^2-2\sum_k a_k^2\sum_k b_k^2+4 \left(\sum_k a_k b_k\right)^2 .
\end{equation}
Then, we have from normalization 
\begin{equation}
1=\sum_k \left(a_k^2+ b_k^2\right)=\left(\sum_k a_k^2+\sum_k b_k^2\right)^2=\left(\sum_k a_k^2\right)^2+\left(\sum_k b_k^2\right)^2+2\sum_k a_k^2\sum_k b_k^2.
\end{equation}

Next, a straightforward calculation yields
\begin{equation}
\tr_{1:n}\left(\ketbra{\psi_\text{R}}{\psi_\text{R}}\right)=\sum_k a_k^2\ket{R}\bra{R}+b_k^2\ket{I}\bra{I}+a_kb_k\ket{R}\bra{I}+a_kb_k\ket{I}\bra{R}\,,
\end{equation}
where $\tr_{1:n}(\cdot)$ is the partial trace over all rebits except the (${n+1}$)-th rebit.
Then, we have 
\begin{equation}
\tr\left(\tr_{1:n}\left(\ketbra{\psi_\text{R}}{\psi_\text{R}}\right)^2\right)=\left(\sum_k a_k^2\right)^2+\left(\sum_k b_k^2\right)^2+2\left(\sum_k a_k b_k\right)^2.
\end{equation}
Now, we can easily see
\begin{equation}\tr\left(\tr_{1:n}\left(\ketbra{\psi_\text{R}}{\psi_\text{R}}\right)^2\right)=\frac{1}{2}(\vert\braket{\psi^*}{\psi}\vert^2+1)\,.
\end{equation}
From this, we immediately conclude
\begin{equation}\label{eq:rebit_imag_sup}
\mathcal{I}(\ket{\psi})\equiv 2(1-\tr\left(\tr_{1:n}(\ketbra{\psi_\text{R}}{\psi_\text{R}})^2)\right).
\end{equation}
We can measure imaginarity via the rebits efficiently by measuring the purity of the flag rebit. Purity of a single rebit can be efficiently measured, e.g. via tomography on the flag rebit which scales as $O(1)$.

\section{Testing for imaginarity, entanglement, magic, coherence and purity}\label{sec:comparison}

Besides imaginarity and coherence, entanglement and magic (also known as nonstabilizerness) are important resources for quantum information processing~\cite{chitambar2019quantum}.
Entanglement, characterized by nonclassical correlations between spatially separated quantum systems, has paved the way for advancements in quantum communication \cite{cacciapuoti2020when, xing2023fundamental}, quantum cryptography \cite{ekert1991quantum}, and quantum computation \cite{briegel2009measurement,biamonte2020entanglement}. On the other hand, magic, which measures the extent of deviation from states producible by Clifford circuits or from transformations effected by such circuits, can be used to perform tasks that are computationally challenging for classical computers \cite{jozsa2014classical,koh2015further,bouland2018complexity,haug2023efficient}. Various measures of magic have been proposed in recent years~\cite{veitch2014resource, bravyi2019simulation, bu2019efficient, seddon2021quantifying,leone2021renyi,bu2022statistical, haug2022scalable,bu2023stabilizer,haug2023efficient,bu2023magic,bu2023quantum}.

Testing for the aforementioned resources is an important task for the certification of quantum information processors. We now discuss the cost of tester $\mathcal{A}_Q$ for different  properties $Q$, where we assume $\beta=0$ 
for the tester according to Def.~\ref{def:prop-tester}. 
Here, we report on the scaling of the best known protocol for testing a representative measure of each property for the case of qubits.
We summarize our findings in Table~\ref{tab:comparison}.

\begin{enumerate}
\item Imaginarity of states cannot be tested efficiently as we have shown in this work, specifically in \SM~\ref{sec:proof_imag}.
In contrast, imaginarity of unitaries can be measured efficiently as shown in Ref.~\cite{huang2021demonstrating} or \SM~\ref{sec:meas_imag_U}.

\item For entanglement, one  can efficiently test for bipartite entanglement of pure states~\cite{ekert2002direct,bendersky2009general}. For unitaries, one can efficiently test whether the unitary is a product unitary~\cite{harrow2013testing}. 

\item The magic of states can be efficiently tested both for states~\cite{gross2021schur,haug2022scalable,haug2023efficient} and unitaries~\cite{low2009learning,wang2011property,gross2021schur}.

\item For states, one can test coherence efficiently using the protocol given in \SM~\ref{sec:coherenceTesting}.
The protocol for the coherence power of unitaries is given in~\cite{zanardi2017coherence} or \SM~\ref{sec:coh_unitary}.

\item Purity of states can be tested via the SWAP test~\cite{barenco1997stabilization}, where the protocol is shown in \SM~\ref{sec:noise}. 
Analogously, for channels one can test how close the channel is to an isometry.
For a channel $\Gamma$, one can efficiently test the property $g(\Gamma)=\text{tr}((\Gamma \otimes I_n (\ket{\Phi})^2)$, where one has $g=1$ only when the channel is an isometry, and else $g<1$. One can efficiently measure $g$ by applying the SWAP test on the Choi state.
\end{enumerate}

All mentioned properties can be efficiently measured via protocols involving Bell measurements, with the exception being the imaginarity of states.

\section{Imaginarity of Haar-random states and unitaries}

Here, we calculate the average imaginarity of both Haar-random states and unitaries \revC{using Weingarten calculus~\cite{collins2006integration}.}

\subsection{Imaginarity of Haar-random states}
\label{sec:imaginarity_Haar_states}

\begin{lemS}
Let $d\in\mathbb{Z}^{+}$ be a positive integer and let $\eta$ denote the Haar measure on $U(\mathbb{C}^{d})$. For $U \in U(\mathbb{C}^{d})$, let $\ket {\psi_U}=U\ket 0$ denote the state obtained by acting the unitary $U$ on the computational basis state $\ket 0$.
Then, 
\begin{align}
\int \d\eta(U)\left[1-\left|\langle\psi_{U}\vert\psi_{U}^{\star}\rangle\right|^{2}\right]=1-\frac{2}{d+1}.
\label{eq:random_imaginarity}
\end{align}
\end{lemS}
\begin{proof}
By definition,
\begin{align*}
\vert\psi_{U}\rangle & =U\vert0\rangle.
\end{align*} 
By transposing or complex-conjugating this state, we get
\begin{align*}
\vert\psi_{U}^{\star}\rangle &= U^{\star}\vert0\rangle, \\
\langle\psi_{U}\vert 
&= \langle0\vert U^{\dagger},\\
\langle\psi_{U}^{\star}\vert 
&= 
\langle0\vert U^{T}.
\end{align*}

These expressions allow us to expand the left-hand side of \eqref{eq:random_imaginarity} as follows:
\begin{align}
\int \d\eta(U)\left[1-\left|\langle\psi_{U}\vert\psi_{U}^{\star}\rangle\right|^{2}\right] & =1-\int \d\eta(U) \ \langle\psi_{U}\vert\psi_{U}^{\star}\rangle\langle\psi_{U}^{\star}\vert\psi_{U}\rangle\nonumber\\
 &= 1-\int \d\eta(U) \ \sum_{a=0}^{d-1}U^{\dagger}_{0a}U^{\star}_{a0}\sum_{b=0}^{d-1}U^{T}_{0b}U_{b0}\nonumber\\
 &= 1-\sum_{a,b=0}^{d-1}\left[\int \d \eta(U) \ U_{b0}U_{b0}U^{\star}_{a0}U^{\star}_{a0}\right].
\label{eq:haar_integral_before_formula}
\end{align}

The integral in \eqref{eq:haar_integral_before_formula} is an integral taken with respect to the Haar measure of a monomial of rank 4; for such an integral, nice closed-form expressions are well-known (for example, see~\cite[Eq.~10]{puchala2017symbolic} or~\cite{collins2006integration}): for $d\geq2$,
\begin{align}
    \int \d\eta(U) \  U_{i_{1}j_{1}}U_{i_{2}j_{2}}U^*_{i_{1}^{\prime}j_{1}^{\prime}}U^*_{i_{2}^{\prime}j_{2}^{\prime}} & =\frac{1}{d^{2}-1}\left(\delta_{i_{1}i_{1}^{\prime}}\delta_{i_{2}i_{2}^{\prime}}\delta_{j_{1}j_{1}^{\prime}}\delta_{j_{2}j_{2}^{\prime}}+\delta_{i_{1}i_{2}^{\prime}}\delta_{i_{2}i_{1}^{\prime}}\delta_{j_{1}j_{2}^{\prime}}\delta_{j_{2}j_{1}^{\prime}}\right)\nonumber\\
 &\quad -\frac{1}{d\left(d^{2}-1\right)}\left(\delta_{i_{1}i_{1}^{\prime}}\delta_{i_{2}i_{2}^{\prime}}\delta_{j_{1}j_{2}^{\prime}}\delta_{j_{2}j_{1}^{\prime}}+\delta_{i_{1}i_{2}^{\prime}}\delta_{i_{2}i_{1}^{\prime}}\delta_{j_{1}j_{1}^{\prime}}\delta_{j_{2}j_{2}^{\prime}}\right).\label{eq:haar_rank_4}
\end{align}

Hence, for $d\geq2$,
\begin{align*}
\int \d \eta(U) U_{b0}U_{b0}U^{\star}_{a0}U^{\star}_{a0} & =\frac{1}{d^{2}-1}\left(\delta_{ba}\delta_{ba}\delta_{00}\delta_{00}+\delta_{ba}\delta_{ba}\delta_{00}\delta_{00}\right) -\frac{1}{d\left(d^{2}-1\right)}\left(\delta_{ba}\delta_{ba}\delta_{00}\delta_{00}+\delta_{ba}\delta_{ba}\delta_{00}\delta_{00}\right)\\
 & =\frac{2}{d\left(d+1\right)}\delta_{ab},
\end{align*}
and so substituting this expression into \eqref{eq:haar_integral_before_formula} gives us
\begin{align}
\int \d\eta(U)\left[1-\left|\langle\psi_{U}\vert\psi_{U}^{\star}\rangle\right|^{2}\right] & =1-\sum_{a,b=0}^{d-1}\frac{2}{d(d+1)}\delta_{ab} =1-\frac{2}{d\left(d+1\right)}d =1-\frac{2}{d+1}.
\end{align}

Next, we show that this expression also holds for $d=1$:
\begin{align*}
\int \d\eta(U) \ \left[1-\left|\langle\psi_{U}\vert\psi_{U}^{\star}\rangle\right|^{2}\right] & =1-\int \d\eta(U) U_{00}U_{b0}U^{\star}_{a0}U^{\star}_{a0}\\
 & =1-\int \d\eta(U) \ \left|U_{00}\right|^{2} =1-\int \d\eta(U) =1-1=0=1-\frac{2}{d+1}.
\end{align*}

This completes the proof of \eqref{eq:random_imaginarity} for all $d\in\mathbb{Z}^{+}$.
\end{proof}

\subsection{Imaginarity of Haar-random unitaries}
\label{sec:imaginarity_Haar_unitaries}

\begin{lemS} \label{lem:haarU}
Let $\mathcal{I}:\mathcal{U}\left(\mathbb{C}^{d}\right)\rightarrow\mathbb{R}$,
$d \in \mathbb{Z}_{\geq 2}$ and $\mathcal{I}\left(U\right)=1-\frac{1}{d^{2}} \vert \tr\left(U^{T}U\right)\vert^{2}$,
then
\[
\int \d\eta(U) \ \mathcal{I}\left(U\right)=1-\frac{2}{d\left(d+1\right)}.
\]
\end{lemS}

\begin{proof}
    \begin{align*}
\int \d\eta(U) \ \mathcal{I}\left(U\right) & =\int \d\eta(U) \left(\frac{1}{d^{2}}\left|\tr\left(U^{T}U\right)\right|^{2}\right)\\
 & =1-\frac{1}{d^{2}}\int \d\eta(U) \ \left|\tr\left(U^{T}U\right)\right|^{2}\\
 & =1-\frac{1}{d^{2}}\int \d\eta(U) \ \tr\left(U^{T}U\right){\tr\left(U^{T}U\right)}^*\\
 & =1-\frac{1}{d^{2}}\int \d\eta(U) \ \tr\left(U^{T}U\right)\tr\left(U^{\dagger}{U}^*\right)\\
 & =1-\frac{1}{d^{2}}\int \d\eta(U) \ \sum_{w\in\mathbb{Z}_{d}}\langle w\vert U^{T}U\vert w\rangle\sum_{y\in\mathbb{Z}_{d}}\langle y\vert U^{\dagger}U^*\vert y\rangle\\
 & =1-\frac{1}{d^{2}}\sum_{w,x,y,z\in\mathbb{Z}_{d}}\int \d\eta(U) \ \langle w\vert U^{T}\vert x\rangle\langle x\vert U\vert w\rangle\langle y\vert U^{\dagger}\vert z\rangle\langle z\vert U^*\vert y\rangle\\
 & =1-\frac{1}{d^{2}}\sum_{w,x,y,z\in\mathbb{Z}_{d}}\int \d \eta (U)\ U_{xw}U_{xw}U^*_{zy}{U}_{zy}^* \\
 & =1-\frac{1}{d^{2}}\sum_{w,x,y,z\in\mathbb{Z}_{d}}\int \d\eta(U) \ U_{xw}^2(U^*_{zy})^2.\numberthis\label{eq:haar_integral_intermediate}
\end{align*}

The integral in \eqref{eq:haar_integral_intermediate} is of the form \eqref{eq:haar_rank_4}, and hence, by taking
\begin{align*}
i_{1} & =i_{2}=x\\
j_{1} & =j_{2}=w\\
i_{1}^{\prime} & =i_{2}^{\prime}=z\\
j_{1}^{\prime} & =j_{2}^{\prime}=y
\end{align*}
in \eqref{eq:haar_rank_4}, we get
\begin{align*}
\int \d\eta(U) \ U_{xw}^2(U^*_{zy})^2 & =\frac{1}{d^{2}-1}\left(\delta_{xz}\delta_{xz}\delta_{wy}\delta_{wy}+\delta_{xz}\delta_{xz}\delta_{wy}\delta_{wy}\right)  -\frac{1}{d\left(d^{2}-1\right)}\left(\delta_{xz}\delta_{xz}\delta_{wy}\delta_{wy}+\delta_{xz}\delta_{xz}\delta_{wy}\delta_{wy}\right)\\
&=  \frac{2}{d\left(d+1\right)}\delta_{xz}\delta_{wy}.
\end{align*}

Thus, \eqref{eq:haar_integral_intermediate} may be simplified as
\begin{align*}
\int \d\eta(U) \ \mathcal{I}\left(U\right) & =1-\frac{1}{d^{2}}\sum_{w,x,y,z\in\mathbb{Z}^{d}}\frac{2}{d\left(d+1\right)}\delta_{xz}\delta_{wy}\\
 & =1-\frac{1}{d^{2}}\frac{2}{d\left(d+1\right)}d^{2}\\
 & =1-\frac{2}{d\left(d+1\right)},
\end{align*}
which completes our proof.
\end{proof}

\end{document}